\documentclass[runningheads]{llncs}

\usepackage{amsfonts,amsmath,amsthm}
\usepackage{url}
\usepackage{hyperref}
\usepackage{todonotes}
\usepackage{algorithm,algpseudocode,booktabs}
\usepackage{pgfplots}
\usepackage{csvsimple}

\usepackage{xspace} %
\usepackage{xparse} %

\newcommand\newmath[2]{\newcommand#1{\ensuremath{#2}\xspace}}
\newcommand\renewmath[2]{\renewcommand#1{\ensuremath{#2}\xspace}}

\newcommand\newmathope[2]{\newcommand#1{\ensuremath{\operatornamewithlimits{#2}}\xspace}}

\usepackage{tikz,pgf}
\usetikzlibrary{arrows,automata,shapes,calc,positioning,intersections,backgrounds}

\newmath{\N}{\mathbb{N}}
\newmath{\Z}{\mathbb{Z}}
\newmath{\Q}{\mathbb{Q}}
\newmath{\R}{\mathbb{R}}

\newmathope{\argmin}{\arg\min}
\newmathope{\argmax}{\arg\max}

\renewmath{\Pr}{\mathbb P}
\newmath{\PP}{\mathbf P}
\newmath{\Bad}{\mathsf{Bad}}
\newmath{\Opt}{\mathsf{Opt}}
\newmath{\Dist}{\mathcal D}
\newmath{\M}{\mathcal M}
\newmath{\Supp}{\mathsf{Supp}}
\newmath{\E}{\mathbb E}
\newmath{\UPre}{\mathsf{UPre}}
\newmath{\Good}{\mathsf{Good}}
\newmath{\Ugly}{\mathsf{Ugly}}

\newmath{\Reward}{\mathsf{Rew}}
\newmath{\AReward}{\mathsf{ARew}}

\newmath{\FPaths}{\mathsf{Paths}}
\newmath{\GoodPaths}{\mathsf{GoodPaths}}
\newmath{\BadPaths}{\mathsf{BadPaths}}
\newmath{\first}{\mathsf{first}}
\newmath{\second}{\mathsf{second}}
\newmath{\last}{\mathsf{last}}
\newmath{\len}{\operatorname{len}}
\newmath{\States}{\mathsf{States}}
\newmath{\Val}{\mathsf{Val}}
\newmath{\Win}{\mathsf{Win}}
\newmath{\ValSafety}{\mathsf{ValSafety}}
\newmath{\ValReach}{\mathsf{ValReach}}
\newmath{\Shield}{\mathsf{Shield}}

\newmath{\opt}{\mathsf{opt}}
\newmath{\Cyl}{\mathsf{Cyl}}
\newmath{\GoodCyl}{\mathsf{GoodCyl}}

\newmath{\NN}{\mathsf{NN}}

\newmath{\reward}{\mathsf{reward}}
\newmath{\numsamples}{\mathsf{count}} %
\newmath{\children}{\mathsf{children}}
\newmath{\mctsvalue}{\mathsf{value}}
\newmath{\total}{\mathsf{total}}
\newmath{\I}{\mathcal{I}}
\newmath{\iter}{\mathsf{iter}}

\newmath{\A}{\mathscr A}
\newmath{\Apsi}{\A^{\psi}}
\newmath{\Aphi}{\A^{\varphi}}

\newmath{\cyl}{\mathsf{Cyl}}
\newmath{\invalpha}{\alpha^{\text{-}1}}

\newcommand{\MP}{\mathsf{MP}}

\allowdisplaybreaks

\usepackage[T1]{fontenc}
\usepackage{fontawesome5}

\usepackage{color}

\begin{document}

\title{Bi-Objective Lexicographic Optimization in Markov Decision Processes with Related Objectives\thanks{G. A. P\'erez was supported by the iBOF ``DESCARTES'' and FWO ``SAILor'' projects. Debraj Chakraborty, Anirban Majumdar, Sayan Mukherjee and Jean-Fran\c{c}ois Raskin were supported by the EOS project \emph{Verifying Learning Artificial Intelligent Systems} (F.R.S.-FNRS and FWO). Debraj Chakraborty was also supported by MASH (MUNI/I/1757/2021) of Masaryk University.
}}
\titlerunning{Bi-Objective Lexicographic Optimization in MDPs}
\author{Damien Busatto-Gaston\inst{1}
\and
 {Debraj Chakraborty}\inst{2}
 \and
 Anirban Majumdar\inst{3}
 \and
 Sayan Mukherjee\inst{3}
 \and
 Guillermo A. P\'erez \inst{4}
 \and
 Jean-Fran\c{c}ois Raskin\inst{3}
 }
\authorrunning{Busatto-Gaston et al.}
\institute{Universit\'{e} Paris Est Créteil, LACL, F-94010, France
    \email{damien.busatto-gaston@u-pec.fr}
 \and
    Masaryk University, Brno, Czech Republic
    \email{chakraborty@fi.muni.cz}
	 \and
 	Universit\'{e} Libre de Bruxelles, Brussels, Belgium
  \email{\{anirban.majumdar,sayan.mukherjee,jean-francois.raskin\}@ulb.be}
  \and
	University of Antwerp -- Flanders Make, Antwerp, Belgium
    \email{guillermo.perez@uantwerpen.be}
}
\maketitle              %
\begin{abstract}
We consider lexicographic bi-objective problems on Markov Decision Processes (MDPs), where we optimize one objective while guaranteeing optimality of another. We propose a two-stage technique for solving such problems when the objectives are related (in a way that we formalize). We instantiate our technique for two natural pairs of objectives: minimizing the (conditional) expected number of steps to a target while guaranteeing the optimal probability of reaching it; and maximizing the (conditional) expected average reward while guaranteeing an optimal probability of staying safe (w.r.t. some safe set of states).
For the first combination of objectives, which covers the classical frozen lake environment from reinforcement learning, we also report on experiments performed using a prototype implementation of our algorithm and compare it with what can be obtained from state-of-the-art probabilistic model checkers solving optimal reachability.

\end{abstract}
\section{Introduction}
\label{sec:intro}

Probabilistic model-checkers, such as {\sc Storm}~\cite{STORM22} or {\sc Prism}~\cite{KNP11}, have been developed to solve the model-checking problem for logics like PCTL and models like Markov decision processes. These tools can be used to compute strategies (or schedulers) that maximize the probability of, for instance, reaching a set of states. As a concrete example, they can be used to solve the Frozen Lake problem shown in Figure~\ref{fig:frozenLake}, where a robot must navigate from an initial point to a target while avoiding holes in the ice.
\begin{figure}[t]
    \centering
    \includegraphics{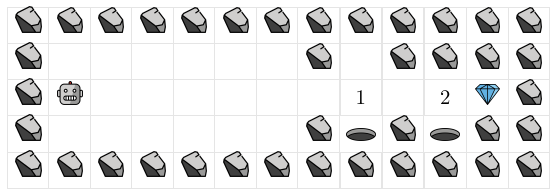}
    \caption{In the game of Frozen Lake, a robot moves in a slippery grid. It
      has to reach the target (the gem) while avoiding holes in the grid. The robot
      can no longer move once in a hole. %
      Part of the grid contains walls and the robot cannot move into them. The frozen
      surface of the lake being slippery,
when the robot tries to move by picking a cardinal direction, the next state
is determined stochastically over adjacent directions. For example, trying to move right would result on the robot going to the cell on the right with probability $0.8$ but going up or down with probability $0.1$ for each.  %
}\label{fig:frozenLake}
 \end{figure}
The ground is frozen  and so the movements of the robot are subject to
stochastic dynamics. While model-checkers provide optimal strategies for the
probability of reaching the target, those strategies may not be efficient in
terms of the expected number of steps required to reach it. For instance, the
strategy returned by \textsc{Storm} for the grid given in
Fig.~\ref{fig:frozenLake} requires on average $345$ steps to reach the target, while %
there are
other strategies that  %
are optimal for reachability
that can reach the target in just $34$ steps on average.
Indeed, a strategy can be optimal in terms of the
   probability to reach the target while (seemingly) behaving like a random walk on the grid (on portions without holes in particular).
 In the worst case, one could expect to reach the target after large number of steps (even on grids where there is a short and direct path to target), which can be considered useless for practical purposes.\footnote{
 In particular, the strategy could be used as a component of some larger approach dealing with a more challenging problem too difficult for exact methods. In these cases, such as~\cite{AAMAS23}, one frequently relies on machine-learning techniques
 (\textit{e.g.}~Monte-Carlo methods or reinforcement learning) that run simulations for a fixed number of steps.
 Thus, a strategy that takes needlessly too many steps to reach a target will not help with learning practical and relevant strategies.}
 Therefore, in this context,
 we aim to not only \emph{maximize the probability} of reaching the target, but also \emph{minimize the expected number of steps} required to reach it, %
 which is thus a multi-objective problem. Unfortunately,
 multi-objective optimization is not yet standard for
 probabilistic model checkers and most of them support it only for specific
 combinations of some objectives. %
 For instance, {\sc
 Storm} can solve the optimal reachability problem and compute the minimal
 expected cost to target, but only for target sets that can be reached with
 probability one. The latter is not usually the case in the Frozen Lake
 problem: the robot may need to walk next to a hole, and risk falling into
 it, in order to reach the target. In this paper, we demonstrate how to
 address the problems we have identified with the Frozen Lake example by
 leveraging the algorithms implemented in {\sc Storm}.

We identify a family of bi-objective optimization problems that can be solved
in two steps using readily available model-checking tools. This family of
problems is formalized as follows. Let $\mathcal{M}$ be an MDP and
$\Sigma(\mathcal{M})$ the set of all strategies for it. We study
reward functions that map strategies $\sigma \in
\Sigma(\mathcal{M})$ to real numbers via the induced MC $\mathcal{M}_\sigma$.
Concretely, let $f, g : \Sigma(\mathcal{M}) \to \mathbb{R}$. We say a strategy
$\sigma \in \Sigma(\mathcal{M})$ is $f$-optimal if $f(\sigma) = \sup_{\tau \in
\Sigma(\mathcal{M})} f(\tau)$ and write $\Sigma_f$ for the set of all
$f$-optimal strategies.

There are multiple ways in which one can approach the problem of finding
optimal strategies with respect to both $f$ and $g$ (see,
e.g.,~\cite{DBLP:conf/stacs/ChatterjeeMH06} and references therein). In this
work, we fix a lexicographic order on the functions. Formally, we want to
compute a strategy $\sigma$ such that the following holds:
\begin{equation}\label{eq:conds}
  \sigma \in \Sigma_f \text{ and } g(\sigma) = \sup_{\tau \in \Sigma_f} g(\tau)
\end{equation}

\paragraph*{Our contribution.}
In this paper, we discuss the problem
described above for two concrete cases of $f$ and $g$.
First,
we tackle the motivating example from Frozen Lake and detail how to find
strategies that maximize $f$, the probability of reaching a set of target
states, while minimizing the conditional expected number of steps to reach them, encoded as
$g$. It is not clear
how to obtain an exact finite representation of the set $\Sigma_f$ of all optimal
strategies for $f$.
To solve this problem, we first compute an over-approximation $\Sigma_f^{\sf over}$ of $\Sigma_f$.
We then prune the original MDP in such a way that the set of all strategies in the pruned MDP is exactly $\Sigma_f^{\sf over}$.
In this context, $\Sigma_f^{\sf over}$ will be the set of strategies that
only play actions used by at least one optimal strategy for reachability.
We then optimize a modified objective $g'$ in the pruned MDP, that, in turn, optimizes both $f$ and $g$ in the lexicographic order in the original MDP. The pruned MDP may contain actions from states that are part of some strategy maximizing the probability of reaching a target but which (taken together) do not make any progress towards the target (for example, a self-loop). These actions, however, are not part of the strategies that optimize $g'$ in the pruned MDP and hence they are also not part of the strategies that are returned by our algorithm.
Secondly, we also consider the problem of maximizing
the probability of remaining in a safe set of states, encoded as $f$, while
maximizing the expected mean-payoff along safe paths, encoded as $g$. Unlike the case for reachability, in this problem, we can in fact construct an exact finite representation of $\Sigma_f$ in the form of an MDP (Theorem~\ref{thm:safetyOpt1}), which we again construct by pruning the original one.
Similar to the reachability case, we then optimize a modified objective $g'$ in the pruned MDP. %
In both of these cases, we prove (in Theorems~\ref{thm:1} and~\ref{thm:3}) that the strategies optimizing $g'$ in the pruned MDP, are solutions to Eq.~\ref{eq:conds}.

Note that, the solution to the second problem
is related to the
\emph{shielding}~\cite{DBLP:conf/aaai/AlshiekhBEKNT18} framework and similar
works~\cite{DBLP:conf/aaai/Chatterjee0PRZ17,DBLP:conf/tacas/Junges0DTK16}, where one computes an exact
representation of the set of all optimal strategies for the first
objective and then solves for the second objective within that space. However, as remarked earlier, it is unclear
how to get an exact representation of $\Sigma_f$ in the first problem.

In both cases, our solution to these (lexicographic) bi-objective problems
can be implemented by using two calls to off-the-shelf tools like \textsc{Storm} or \textsc{PRISM}, thus
resulting in a
polynomial-time solution.
We
report on experimental results for the Frozen Lake example that validate the
need and practicality of our approach. Finally, we discuss
other instances of multi-objective problems where our approach naturally
generalizes.

\paragraph*{Related works.}
 The strategy synthesis problem in MDPs (or stochastic games, their $2.5$-players extension) can be defined %
for a wide variety of temporal objectives and quantitative rewards. %
Multi-objective problems are particularly challenging, as they need to optimize for multiple, potentially conflicting, goals.
\cite{Chatterjee2023} detailed a strategy synthesis algorithm for lexicographic combinations of $\omega$-regular objectives.
This problem has also been studied with model-free, reinforcement learning approaches~\cite{10.1007/978-3-030-90870-6_8}.
However, these approaches do not consider objectives that maximize quantitative rewards, and cannot optimize for properties such as the time to reach a target.
In \cite{10.1007/978-3-642-36742-7_12}, one can mix LTL objectives with mean-payoff rewards and in \cite{ijcai2022p476}
a lexicographic combination of discounted-sum rewards is considered.
Moreover, a discounted semantics of LTL\footnote{This allows one to express constraints on the number of steps needed to satisfy an Until operator.} is studied in \cite{10.1007/978-3-642-54862-8_37}, and can be used as a way to optimize for the time until a target is reached.
Combinations of LTL and total-reward objectives have been considered in works such as \cite{10.1007/978-3-642-19835-9_11} and \cite{10.1007/978-3-642-40196-1_28}, under assumptions that exclude our problem.
Indeed, while minimizing the time to reach a target can be encoded as optimizing the total-reward of a slightly modified structure (where costs are $1$ at every move before the target is reached then $0$ forever), these works are not directly applicable to our problem: applying \cite{10.1007/978-3-642-19835-9_11} requires the assumption that the optimal probability to reach a target is $1$ in order to minimize the expected time to target, and \cite{10.1007/978-3-642-40196-1_28} searches for a strategy on the Pareto frontier instead of optimizing for a lexicographic combination of objectives.

Note that minimizing the time to reach a target (a variant of the \emph{stochastic shortest path problem}~\cite{DBLP:journals/mor/BertsekasT91}) is only well-defined under the condition that the target is reached,
so that our example requires studying \emph{conditional} probabilities.
This notion has been studied in single-objective settings, so that for example
probabilistic model-checkers can optimize for the (conditional) probability of satisfying an $\omega$-regular event under the condition that another $\omega$-regular event holds~\cite{10.1007/978-3-642-54862-8_43}.
In particular, \cite{10.1007/978-3-662-54580-5_16} details how to maximize the expected total-reward until a target is reached, under the assumption that the target is indeed reached with positive probability.
This does not solve our motivating example however, as it may yield a strategy that is suboptimal for the probability of reaching the target.
Finally, we note that tools such as \cite{10.1007/978-3-642-19835-9_24} can handle settings similar to our second example (optimizing for safety and mean-payoff), but they do not consider conditional mean-payoff.

Overall, our general two-stage technique covers combinations of objectives that are subcases of problems previously studied (e.g. in~\cite{Chatterjee2023}) but it is also applicable to combinations not previously considered. Interestingly, and to the best of our knowledge, optimizing for a reachability objective while minimizing the conditional time to satisfy is not formally covered by previous work on multi-objective strategy synthesis, and is not an available feature of probabilistic model-checkers. It may be possible that this problem can be reduced to finding \emph{bias-optimal strategies}~\cite{dede287b-e396-3536-a645-5de6f58195ff} in a slightly modified MDP. However, this does not generalize to other objectives.

\section{Preliminaries}

A \emph{probability distribution} on a countable set $S$ is a function $d:S\to [0,1]$ such that $\sum_{s\in S}d(s)=1$.
We denote the set of all probability distributions on set $S$ by $\Dist(S)$. The support of a distribution $d\in \Dist(S)$ is $\Supp(d)=\{s\in S\mid d(s)>0\}$.

\subsection{Markov Chain}
\begin{definition}[Markov chain]\label{def:mc}
	A (discrete-time) Markov chain or an MC is a tuple $M=(S,P)$, where
	 $S$ is a countable set of states and
		$P$ is a mapping from $S$ to $\Dist(S)$.
\end{definition}
For states $s,s'\in S$, $P(s)(s')$ denotes the probability of moving from state $s$ to state $s'$ in a single transition and we denote this probability $P(s)(s')$ as $P(s,s')$.

For a Markov chain $M$,
a \emph{finite path} $\rho = s_0s_1\ldots s_i$ of length $i>0$ is
a sequence of $i+1$ consecutive states such that for all $t\in[0,i-1]$, $s_{t+1}\in \Supp(P(s_t))$.
We also consider states to be paths of length $0$.
Similarly, An \emph{infinite path} is an infinite sequence $\rho = s_0s_1s_2\ldots$ of states such that for all $t\in\N$, $s_{t+1}\in \Supp(P(s_t))$.
For a finite or infinite path $\rho=s_0s_1\ldots$, we denote its $(i+1)^{th}$ state by $\rho[i] = s_i$. We denote the last state of a finite path
$\rho = s_0s_1\ldots s_n$ by $\last(\rho) = s_n$.
Let $\rho=s_0s_1\ldots s_i$ and $\rho'=s'_0s'_1\ldots s'_j$ be two paths
such that $s_i=s'_0$.
Then, $\rho\cdot \rho'$
denotes the path $s_0s_1\ldots s_is'_1\ldots s'_j$.
For a finite or infinite path $\rho=s_0s_1\ldots$, we denote its \emph{$i$-length prefix} as $\rho_{|_i} = s_0s_1\ldots s_i$.

For a finite  path $\rho\in\FPaths_{M}$, we use $\FPaths^{\omega}_{M}(\rho)$ to denote the set of all paths $\rho' \in \FPaths^{\omega}_{M}$ such that there exists $\rho''\in\FPaths^{\omega}_M$ with $\rho'=\rho\cdot \rho''$. $\FPaths^{\omega}_{M}(\rho)$ is called the \emph{cylinder set} of $\rho$.

The $\sigma$-algebra associated with the MC $M$ %
is the smallest $\sigma$-algebra that contains the cylinder sets $\FPaths_{M}^{\omega}(\rho)$ for all $\rho\in \FPaths_{M}$. For a state $s$ in $S$, a measure is defined for the cylinder sets as --
\begin{align*}
	\Pr_{M,s}(\FPaths_{M}^{\omega}(s_0s_1\ldots s_i))=&\begin{cases}
		\prod_{t=0}^{i-1}P(s_t)(s_{t+1}) &\text{if } s_0 = s\\
		0 &\text{otherwise.}
	\end{cases}
\end{align*}

We also have $\Pr_{M,s}(\FPaths_{M}^{\omega}(s)) = 1$ and $\Pr_{M,s}(\FPaths_{M}^{\omega}(s')) = 0$ for $s'\neq s$.
This can be extended to a unique probability measure $\Pr_{M,s}$ on the aforementioned $\sigma$-algebra. In particular, if $\mathcal C\subseteq\FPaths_{M}$ is a set of finite paths forming pairwise disjoint cylinder sets, then
$\Pr_{M,s}(\cup_{\rho\in \mathcal C}\FPaths_{M}^{\omega}(\rho))=
\sum_{\rho\in \mathcal C}\Pr_{M,s}(\FPaths_{M}^{\omega}(\rho))$.
Moreover, if $\Pi\in\FPaths^{\omega}_{M}$ is the complement of a measurable set $\Pi'$, then $\Pr_{M,s}(\Pi)=1-\Pr_{M,s}(\Pi')$.

\subsection{Markov Decision Process}
\begin{definition}[Markov decision process]\label{def:mdp}
	A Markov decision process or an MDP is a tuple $\M=(S,A,P)$, where
	 $S$ is a finite set of states,
	 $A$ is a finite set of actions, and
	 $P$ is a (partial) mapping from $S\times A$ to $\Dist(S)$. %
\end{definition}
$P(s,a)(s')$ denotes the probability that action $a$ in state $s$ leads to state $s'$ and we denote this probability $P(s,a)(s')$ as $P(s,a,s')$.
Note that not all actions may be \emph{legal} from a state. Therefore, if an action $a$ is legal from a state $s$, we will have $\sum_{s'\in S} P(s,a,s')=1$. Otherwise, we will have $P(s,a,s')$ is undefined (denoted by $\bot$) for all $s'\in S$.

The definitions and notations used for paths in Markov chain can be extended in the case of MDPs. In an MDP, a \emph{path} is
a sequence of states and actions.

For an MDP $\M$, a (probabilistic) \emph{strategy} is a function $\sigma : \FPaths_{\M} \to \Dist(A)$
that maps a finite path $\rho$ to a probability distribution in $\Dist(A)$.
For a path $\rho \in \FPaths_\M$ and a strategy $\sigma$,
we will write $\sigma(\rho,a)$ in place of $\sigma(\rho)(a)$.
A strategy $\sigma$ is \emph{deterministic} if the support of
the probability distributions $\sigma(\rho)$ has size $1$.
A strategy $\sigma$ is \emph{memoryless} if $\sigma(\rho)$ depends only on $\last(\rho)$,
i.e. if $\sigma$ satisfies that for all $\rho,\rho'\in\FPaths_{\M}$, $\last(\rho)=\last(\rho')\Rightarrow\sigma(\rho)=\sigma(\rho')$.
We denote the set of all finite paths in $\M$ starting from $s$  following $\sigma$ by  $\FPaths_{\M}(s,\sigma)$.

An MDP $\M$ induced by a strategy $\sigma$ defines an MC $\M_{\sigma}$. Intuitively, this is obtained by unfolding $\M$ using the strategy $\sigma$ and using the probabilities in $\M$ to define the transition probabilities. Formally, $\M_{\sigma} = (\FPaths_{\M},P_{\sigma})$ where for all paths $\rho\in \FPaths_{\M}$, $P_{\sigma}(\rho)(\rho\cdot as) = \sigma(\rho)(a)\cdot P(\last(\rho),a)(s)$.
Thus, a state $\rho$ in $\FPaths_{\M}$ uniquely \emph{matches} a finite path $\rho'$ in $\M_{\sigma}$ where $\last(\rho') = \rho$. This way when a strategy $\sigma$ and a state $s$ is fixed, the probability measure $\Pr_{\M_{\sigma},s}$ defined in $\M_{\sigma}$ is also extended for paths in $\FPaths_\M$.
We write the expected value of a random variable $X$ with respect to the probability distribution $\Pr_{\M_{\sigma},s}$ as $\E_{\M_{\sigma},s}(X)$.
For the ease of notation, we write $\Pr_{\M_{\sigma},s}$ and $\E_{\M_{\sigma},s}$ as $\Pr_{\sigma,s}$ and $\E_{\sigma,s}$ respectively, if the MDP $\M$ is clear from the context. Also, we write $\FPaths^{\omega}_{\M_{\sigma}}(\rho)$ as $\Cyl_{\sigma}(\rho)$, if the MDP $\M$ is clear from the context.

In the sequel, we make use of (technical) lemmas that follow from the extensive literature on Markov chains and MDPs. However, for completeness, and to give the reader intuition regarding the presented objectives, we also give proofs for some of them.
\section{Length-Optimal Strategy for Reachability}
\label{sec:reachability-expected-length}

We begin by considering the multi-objective problem motivated by the game of frozen lake -- the robot tries to reach a target with as few steps as possible while not compromising on the probability of reaching a target.
More formally,
in this section, we find a strategy in an MDP that minimizes the expected number of steps to reach some goal states among those strategies that maximize the probability of reaching the goal states. %

We consider a set of target states $T \subseteq S$ in $\M$, and assume that every state in $T$ is a sink state, that is, it has only one outgoing action to itself with probability~$1$.
Given a path $\rho$ in an MC $M=(S,P)$, we use $\len_T(\rho)$ to denote the length of the shortest prefix of $\rho$ that reaches one of the states of $T$, that is, $\len_T(\rho) = i$ if $\rho[i]\in T$ and for all $j<i$, $\rho[j]\notin T$.

For an MDP $\M = (S,A,P)$,
let $\Pr_{\sigma,s}(\Diamond T)$ be the probability of reaching a state in $T$, starting from $s \in S$, following the strategy $\sigma$ in $\M$.
Then, let $\Val_\M(s) = \max_\sigma \Pr_{\sigma, s}(\Diamond T)$ be the maximum probability to reach $T$ from $s$, and $\Sigma_{\M,s}(\Diamond T) = \argmax_{\sigma}\Pr_{{\sigma},s}(\Diamond T)$ be the set of all optimal strategies.  %

\subsection*{Problem statement.}
Given an MDP $\M$, an initial state $s_0$ and a set of goal states $T$,
our objective is to find a strategy that minimizes %
$\E_{{\sigma},s_{0}}(\len_T \mid \Diamond T)$ among the strategies in $\Sigma_{\M,s_0}(\Diamond T)$, that is, the strategies which maximize $\Pr_{{\sigma},s_0}(\Diamond T)$.

For the rest of this section, we fix the MDP $\M = (S, A, P)$ and a set of target states $T \subseteq S$.
Note that, in this case, the functions $\sigma\mapsto \Pr_{{\sigma},s_0}(\Diamond T)$ and $\sigma\mapsto -\E_{{\sigma},s_0}(\len_T \mid \Diamond T)$ correspond to the two functions $f$ and $g$, respectively, and the set $\Sigma_{\M,s_0}(\Diamond T)$ corresponds to $\Sigma_f$, described in the introduction (Eq.~\ref{eq:conds}).

\subsection{Maximizing Probability to Reach a Target}

We denote the set $\{(s,a)\in S\times A\mid \Val_{\M}(s) = \sum_{s'}P(s,a,s')\cdot \Val_{\M}(s')\}$ by  $\Opt_{\M}$. %
  For $s\in S$, let $\Opt_{\M}(s)$ be the set $\{a \mid (s,a)\in \Opt_{\M}\}$.
Finally, we use $\Sigma^{\Opt}_{\M}$ to represent the set of strategies that takes actions according to $\Opt_{\M}$, that is,
$\Sigma^{\Opt}_{\M} = \{\sigma\mid \forall \rho, \forall a\in\Supp(\sigma(\rho));(\last(\rho),a)\in \Opt_{\M}\}\,.$

\begin{lemma}\label{lemma:case1}
	For every state $s \in S$ and for every $a \in A$,
	\[\Val_{\M}(s)\geq \sum_{s'}P(s,a,s')\cdot \Val_{\M}(s')\,.\]
\end{lemma}
\begin{proof}
	Suppose, there is a state $s\in S$ and an action $a\in A$ such that $\Val_{\M}(s)< \sum_{s'}P(s,a,s')\cdot \Val_{\M}(s')$.
Now, consider the strategy $\sigma'$ that takes action $a$ from $s$ and then from paths $s\cdot as'$ follows a strategy $\sigma_{s'}\in \Sigma_{\M,s'}(\Diamond T)$ that maximizes the probability to reach states in $T$ from $s'$. Formally,
\begin{align*}
	\sigma'(\rho)=&\begin{cases}
		a &\text{if } \rho = s\\
		\sigma_{s'}(\rho') &\text{if } \rho = s\cdot as'\cdot \rho'%
	\end{cases}
\end{align*}
Then,
$\Pr_{{\sigma'},s}(\Diamond T) = \sum_{s'}P(s,a,s')\cdot \Pr_{{\sigma_{s'}},s'}(\Diamond T) = \sum_{s'}P(s,a,s')\cdot \Val_{\M}(s') > \Val_{\M}(s)\,,   $
which is a contradiction as $\Val_{\M}(s)\geq \Pr_{{\sigma},s}(\Diamond T)$ for any $\sigma$.
\end{proof}

\begin{lemma}\label{lem:localOptCoverReach}
	For every state $s \in S$, $\Sigma_{\M,s}(\Diamond T)\subseteq \Sigma^{\Opt}_{\M}$.
\end{lemma}

\begin{proof}
Towards a contradiction, suppose that, there is a strategy $\sigma^*\in \Sigma_{\M,s}(\Diamond T)$ such that $\sigma^* \notin \Sigma_\M^{\Opt}$. Then there exists a path $\rho$ and an action $a\in \Supp(\sigma^*(\rho))$ such that $(\last(\rho),a)\not\in \Opt_{\M}$. Let $\last(\rho) = t$.
Then, from Lemma~\ref{lemma:case1} and the fact that $(t,a)\not\in \Opt_{\M}$, we get:
\begin{equation}\label{eq:valReach1}
\Val_{\M}(t)> \sum_{s'}P(t,a,s')\cdot \Val_{\M}(s')\,,
\end{equation}
and for every other action $a'\neq a$,
\begin{equation}\label{eq:valReach2}
\Val_{\M}(t)\geq \sum_{s'}P(t,a',s')\cdot \Val_{\M}(s')\,.
\end{equation}
Consider the strategy $\overline{\sigma}^*$ that differs from $\sigma^*$ only on paths with $\rho$ as prefix: on every path having $\rho$ as a prefix, $\overline{\sigma}^*$ takes the next action according to a strategy $\sigma_{t}\in \Sigma_{\M,t}(\Diamond T)$ that maximizes the probability to reach a state in $T$ from $t$, whereas, it takes action according to $\sigma^*$ on every other path. Formally,
\begin{align*}
	\overline{\sigma}^*(\rho')=&\begin{cases}
		\sigma_{t}(\rho'') &\text{if } \rho' = \rho\cdot \rho''\\
		\sigma^*(\rho') &\text{otherwise.}
	\end{cases}
\end{align*}
Note that, for every strategy $\sigma$,
and for all $a' \in A$,
\(
\Pr_{\M_{\sigma^*},\rho\cdot a's'}(\Diamond T) \leq \Val_{\M}(s')\,.
\)
\small
\begin{align*}
\text{Therefore, \quad}
	\Pr_{{\sigma^*},\rho}(\Diamond T)
	&= \sum_{a'}\bigg(\sigma^*(\rho,a')
	\cdot \sum_{s'}\left(P(t, a',s')\cdot \Pr_{{\sigma^*},\rho\cdot a's'}(\Diamond T)\right)\bigg)\\
	&\leq \sum_{a'}\bigg(\sigma^*(\rho,a')
	\cdot \sum_{s'}\left(P(t, a',s')\cdot \Val_{\M}(s')\right)\bigg)\\
	&< \sum_{a'}\sigma^*(\rho,a')\cdot
	\Val_{\M}(t)\quad [\text{from Eq.~\ref{eq:valReach1} and \ref{eq:valReach2}}]\\
	&=\Val_{\M}(t)
\end{align*}
\normalsize
So, $\Pr_{{\overline{\sigma}^*},\rho}(\Diamond T) = \Pr_{{\sigma_{t}},t}(\Diamond T) = \Val_{\M}(t) > \Pr_{{\sigma^*},\rho}(\Diamond T)$.
For a finite path $\rho$ and an infinite path $\rho'$, we write $\rho\sqsubseteq \rho'$ if there exists an infinite path $\rho''$ such that $\rho'=\rho\cdot\rho''$.
Now note that, for every strategy $\sigma$,
\begin{align}\label{eq:prefixReach}
 \Pr_{{\sigma},s}(\Diamond T) &= \Pr_{{\sigma},s}(\rho'\models\Diamond T\wedge \rho\sqsubseteq \rho') + \Pr_{{\sigma},s}(\rho'\models\Diamond T\wedge \rho\not\sqsubseteq \rho')\nonumber\\
	&= \Pr_{{\sigma},s}(\Cyl_{\sigma}(p))\cdot \Pr_{{\sigma},\rho}(\Diamond T) + \Pr_{{\sigma},s}(\rho'\models\Diamond T\wedge \rho\not\sqsubseteq \rho')
\end{align}
Since for any $\rho'$ such that $\rho \not \sqsubseteq \rho'$,  $\sigma^*(\rho') = \overline{\sigma}^*(\rho')$, we have $\Pr_{{\sigma^*},s}(\Cyl_{\sigma}(\rho))$ is equal to $\Pr_{{\overline{\sigma}^*},s}(\Cyl_{\overline{\sigma}^*}(\rho))$,
and furthermore, $\Pr_{{\sigma^*},s}(\rho'\models\Diamond T\wedge \rho\not\sqsubseteq \rho')$ is equal to $\Pr_{{\overline{\sigma}^*},s}(\rho'\models\Diamond T\wedge \rho\not\sqsubseteq \rho')$.
Plugging this into Eq.~\ref{eq:prefixReach} for $\sigma^*$ and $\overline{\sigma}^*$, and the fact that $\Pr_{{\sigma^*},\rho}(\Diamond T) < \Pr_{{\overline{\sigma}^*},\rho}(\Diamond T)$, we conclude $\Pr_{{\sigma^*},s}(\Diamond T)<\Pr_{{\overline{\sigma}^*},s}(\Diamond T)$, which contradicts the fact that $\sigma^*$ is an optimal strategy.
\end{proof}

\subsection{Minimizing Expected Conditional Length to Target}
In the following, we propose a simple two-step pruning
algorithm to solve the multi-objective problem defined earlier in this section. Towards that direction, we first modify the given MDP $\M$ in the following manner.
\begin{definition}
\label{def:mprime-from-m}
We define the \emph{pruned MDP} $\M' = (S',A,P')$ with
$S' = \{s\in S\mid \Val_{\M}(s)> 0\}$ and $P'$ constructed from $P$ in the following way:
\begin{align*}
	P'(s,a,s')=&\begin{cases}
		P(s,a,s')\cdot\frac{\Val_{\M}(s')}{\Val_{\M}(s)} &\text{if } (s,a) \in \Opt_{\M}\text{ and }s, s'\in S'\\
		\bot &\text{otherwise.}
	\end{cases}
\end{align*}
\end{definition}
Note that $\M' = (S',A,P')$ is well-defined, since $P'$ is  a probability distribution. Indeed,
	$\sum_{s'}P'(s,a,s') = \sum_{s'}P(s,a,s')\cdot\frac{\Val_{\M}(s')}{\Val_{\M}(s)}
		= \frac{\Val_{\M}(s)}{\Val_{\M}(s)}
		= 1$.

From the construction of $\M'$, we get that the set $\Sigma(\M')$ of all strategies in $\M'$ is, in fact, $\Sigma^{\Opt}_{\M}$.
Following similar notation as introduced earlier, for a strategy $\sigma\in \Sigma(\M')$, we write $\Pr_{\M'_{\sigma},s}$  and $\E_{\M'_{\sigma},s}$ as $\Pr'_{\sigma,s}$ and $\E'_{\sigma,s}$, respectively. Also, we write $\FPaths^{\omega}_{\M'_{\sigma}}(\rho)$ as $\Cyl'_{\sigma}(\rho)$.

We now have all the ingredients to present the algorithm:

\begin{algorithm}[h]
	\caption{}\label{alg:distMin}
	\begin{algorithmic}
		\Require $\M = (S,A,P)$, $s_0 \in S$ and $T\subseteq S$.
	\end{algorithmic}
	\begin{algorithmic}[1]
		\State Create MDP $\M' = (S',A,P')$ according to Definition~\ref{def:mprime-from-m}.
		\State Find a strategy $\sigma^*$ that minimizes the expected length in $\M'$:
		$$ \textstyle \sigma^* \in \argmin_\sigma \E'_{\sigma,s_0}(\len_T)\,.$$\\
		\Return $\sigma^*$.
	\end{algorithmic}
\end{algorithm}

Note that, the strategies present in the pruned MDP $\M'$ contain every strategy of $\M$ that optimizes the probability of reaching a target (Lemma~\ref{lem:localOptCoverReach}).
To show that Algorithm~\ref{alg:distMin} indeed returns a length-optimal strategy maximizing the probability of reachability in $\M$,
we need to show the following:
\begin{itemize}
	\item the strategy given by Algorithm~\ref{alg:distMin} is indeed a strategy that optimizes the probability to reach a target, and
	\item for every strategy $\sigma\in \Sigma_{\M,s_0}(\Diamond T)$, the conditional expected length to a target state $\E_{\sigma,s_0}(\len_T\mid  \Diamond T)$ in $\M$ is the same as $\E'_{\sigma,s_0}(\len_T)$ in $\M'$. Therefore, it is enough to minimize the expected length
    in $\M'$.
\end{itemize}
 We first show a relation between the measures of cylinder sets in $\M$ and $\M'$.
\begin{lemma}\label{lemma:cylRelation}
For every strategy %
	$\sigma\in\Sigma^{\Opt}_{\M}$
    and for every path
	$\rho=s_0a_0s_1\ldots s_n\in {((S'\setminus T) \cdot A)^*T}\cap \FPaths_{\M'}(s_0,\sigma)$,
	$\Pr'_{\sigma,s_0}(\Cyl'_{\sigma}(\rho)) = \frac{\Pr_{\sigma,s_0}(\Cyl_{\sigma}(\rho))}{\Val_{\M}(s_0)}\,.$
\end{lemma}
\begin{proof}
As $s_n\in T$, $\Val_{\M}(s_n) = 1$. So,
     \small
	\begin{align*}
		\Pr'_{\sigma,s_0}(\Cyl'_{\sigma}(\rho))
		& = \prod_{i = 0}^{n-1} \sigma(\rho_{|i},a_i)\cdot P'(s_i,a_i,s_{i+1})\\
		& = \prod_{i = 0}^{n-1} \sigma(\rho_{|i},a_i)\cdot P(s_i,a_i,s_{i+1})\cdot \frac{\Val_{\M}(s_{i+1})}{\Val_{\M}(s_i)}\\
		& = \Pr_{\sigma,s_0}(\Cyl_{\sigma}(\rho))\cdot \frac{\Val_{\M}(s_n)}{\Val_{\M}(s_0)}
        = \frac{\Pr_{\sigma,s_0}(\Cyl_{\sigma}(\rho))}{\Val_{\M}(s_0)}\,.\qedhere
	\end{align*}
  \normalsize
\end{proof}

Using Lemma~\ref{lemma:cylRelation} we will prove that (cf. Corollary~\ref{corr:len}) every strategy that maximizes the probability of reaching a target state in $\M$, reaches a target state in $\M'$ with probability~$1$, and vice versa.

\begin{lemma}\label{lemma:reachabiltyMax}
 For every strategy $\sigma\in\Sigma^{\Opt}_{\M}$, %
	$\Pr'_{\sigma,s_0}(\Diamond T) = \frac{\Pr_{\sigma,s_0}(\Diamond T)}{\Val_{\M}(s_0)}$.
\end{lemma}
\begin{proof}
	Note that, $\FPaths_{\M'}(s_0)\cap ((S'\setminus T)\cdot A)^*T = \FPaths_{\M}(s_0)\cap ((S\setminus T)\cdot A)^*T$, since in the construction of $\M'$ we only remove states of $\M$ from which no state in $T$ is reachable. Therefore, using  Lemma~\ref{lemma:cylRelation}, we get:
    \small
	\begin{align*}
		\Pr'_{\sigma,s_0}(\Diamond T)
		& = \sum_{\rho\in\FPaths_{\M'}(s_0)\cap ((S'\setminus T)A)^*T}
		\Pr'_{\sigma,s_0}(\Cyl'_{\sigma}(\rho))\\
		& = \sum_{\rho\in\FPaths_\M(s_0)\cap ((S\setminus T)A)^*T}
		\frac{\Pr_{\sigma,s_0}(\Cyl_{\sigma}(\rho))}{\Val_{\M}(s_0)}\\
		& = \frac{\Pr_{\sigma,s_0}(\Diamond T)}{\Val_{\M}(s_0)}\,.\qedhere
	\end{align*}
    \normalsize
\end{proof}

\begin{corollary}\label{corr:len}
	For %
 every $\sigma\in\Sigma(\M)$, $\sigma \in \Sigma_{\M,s_0}(\Diamond T)$ iff $\Pr'_{\sigma,s_0}(\Diamond T) = 1$.
\end{corollary}

\noindent Since for every $\sigma \in \Sigma_{\M,s_0}(\Diamond T)$, $\Pr_{\sigma,s_0}(\Diamond T \mid \Diamond T) = 1$, we can write:
\small
\begin{align*}
\E_{\sigma,s_0}(\len_T\mid  \Diamond T)
&= \sum_{r=0}^{\infty}r \cdot
		\frac{\Pr_{\sigma,s_0}(\{\rho \mid \rho \models \Diamond T \land \len_T(\rho) = r\})}{\Pr_{\sigma,s_0}(\Diamond T)} \\
& = \sum_{r=0}^{\infty}r \cdot \sum_{\rho \in \FPaths_{\M}(s_0)\cap ((S\setminus T)A)^*T : \len_T(\rho) = r}
    \frac{\Pr_{\sigma,s_0}(\Cyl_\sigma(\rho))}{\Val_\M(s_0)}\,.
\end{align*}
\normalsize

We now relate the expected length of reaching a target state in $\M'$ with the expected conditional length of reaching a target state in $\M$.

\begin{lemma}\label{lemma:expectationMin}
	For any strategy $\sigma\in\Sigma_{\M,s_0}(\Diamond T)$,
	$\E'_{\sigma,s_0}(\len_T) =
\E_{\sigma,s_0}(\len_T\mid  \Diamond T)$.
\end{lemma}
\begin{proof}
Using $\FPaths_{\M'}(s_0)\cap ((S'\setminus T)\cdot A)^*T = \FPaths_{\M}(s_0)\cap ((S\setminus T)\cdot A)^*T$, Lemma~\ref{lemma:cylRelation} and Corollary~\ref{corr:len}, we get:
 \small
	\begin{align*}
		\E'_{\sigma,s_0}(\len_T)
		& = \sum_{r=0}^{\infty}r \cdot
		{\Pr'_{\sigma,s_0}(\{\rho \mid \rho \models \Diamond T \land \len_T(\rho) = r\})}\\
		& = \sum_{r=0}^{\infty}r \cdot \sum_{\rho \in \FPaths_{\M'}(s_0)\cap ((S'\setminus T)A)^*T : \len_T(\rho) = r}
		{\Pr'_{\sigma,s_0}(\Cyl'_\sigma(\rho))}\\
		& = \sum_{r=0}^{\infty}r \cdot \sum_{\rho \in \FPaths_{\M}(s_0)\cap ((S\setminus T)A)^*T : \len_T(\rho) = r}
		\frac{\Pr_{\sigma,s_0}(\Cyl_\sigma(\rho))}{\Val_{\M}(s_0)} \\%
		& = {\E_{\sigma,s_0}(\len_T\mid  \Diamond T)}\,.\qedhere
	\end{align*}
 \normalsize
\end{proof}

Finally, we prove the correctness of Algorithm~\ref{alg:distMin}:
\begin{theorem}\label{thm:1}
 Given an MDP $\M = (S, A, P)$, a state $s_0 \in S$ and $T \subseteq S$,
 let $\sigma^*$ be the strategy  returned by Algorithm~\ref{alg:distMin}. %
 Then,
	\begin{enumerate}
		\item $\Pr_{\sigma^*,s_0}(\Diamond T) = \Val_{\M}(s_0)$.
		\item $\displaystyle \E_{\sigma^*,s_0}(\len_T\mid\Diamond T) =
		\min_{\sigma\in \Sigma_{\M,s_0}(\Diamond T)}\E_{\sigma,s_0}(\len_T\mid\Diamond T)$
	\end{enumerate}
\end{theorem}
\begin{proof}
	From Corollary~\ref{corr:len}, we  get that $\E'_{\sigma,s_0}(\len_T)\neq \infty$ iff ${\sigma\in\Sigma_{\M,s_0}(\Diamond T)}$.
 So if $\sigma^*\notin\Sigma_{\M,s_0}(\Diamond T)$, then $\E'_{\sigma^*,s_0}(\len_T)=\infty$.
 But since %
 for any strategy $\sigma$ in $\Sigma_{\M,s_0}(\Diamond T)$, $\E'_{\sigma,s_0}(\len_T)<\infty$, it contradicts the fact that $\sigma^*$ minimizes $\E'_{\sigma,s_0}(\len_T)$.
  Therefore, $\sigma^*\in\Sigma_{\M,s_0}(\Diamond T)$, and hence  $\Pr_{\sigma^*,s_0}(\Diamond T) = \Val_{\M}(s_0)$.

	From Lemma~\ref{lemma:expectationMin}, we get  for any $\sigma \in \Sigma_{\M,s_0}(\Diamond T)$,
	\begin{align*}
 \E_{\sigma,s_0}(\len_T\mid  \Diamond T)
	&=\E'_{\sigma,s_0}(\len_T)\\
 \implies \argmin_{\sigma\in \Sigma_{\M,s_0}(\Diamond T)} \E_{\sigma,s_0}(\len_T\mid  \Diamond T)
	&=\argmin_{\sigma\in \Sigma_{\M,s_0}(\Diamond T)} \E'_{\sigma,s_0}(\len_T)
\end{align*}
Hence,  $\sigma^*\in %
\argmin_{\sigma\in \Sigma_{\M,s_0}(\Diamond T)}\E_{\sigma,s_0}(\len_T\mid  \Diamond T)$ and therefore,
we conclude, $\displaystyle \E_{\sigma^*,s_0}(\len_T\mid\Diamond T) =
\min_{\sigma\in \Sigma_{\M,s_0}(\Diamond T)}\E_{\sigma,s_0}(\len_T\mid\Diamond T)$.
\end{proof}

Note that, constructing the MDP $\M'$ (Line~1 of Algorithm~\ref{alg:distMin}) takes polynomial time. Finding a strategy that optimizes $\E'_{\sigma,s_0}(\len_T)$ also takes polynomial time~\cite{DBLP:journals/mor/BertsekasT91}. Therefore, the overall algorithm terminates in polynomial time.

\section{Experimental Results}

We have made a prototype implementation of the pruning-based algorithm (Algorithm~\ref{alg:distMin}) described in Section~\ref{sec:reachability-expected-length}. In this section, we compare the performance (expected number of steps to reach the goal states) of our algorithm with the strategies generated by {\sc Storm} that (only) maximize the probability of reaching the goal states.

In our MDP, when the robot tries to move by picking a direction, the next state is determined randomly over the neighbouring positions of the robot, according to the following distribution weights: the intended direction gets a weight of $10$, and other directions that are not a wall and not the reverse direction of the intended one get a weight of $1$, the distribution is then normalized so that the weights sum up to $1$.

We generated $100$ layouts of size $10\times 10$ where we placed walls in (i) each cell in the border of the grid and (ii) with probability $0.1$, at each of other cells. We then placed holes in the remaining empty cells with the same probability. Finally, we chose the position of the target and the starting position from the remaining empty cells uniformly at random.

From these layouts, we constructed MDPs described in the {\sc Prism} language, a format supported by {\sc Storm}.
For each MDP, we extracted two strategies: (i) a strategy $\sigma_{\mathrm{Storm}}\in \Sigma_{\M,s}(\Diamond T)$ that is produced by {\sc Storm} that optimizes the probability to reach the target, and
(ii) $\sigma_{\mathrm{DistOpt}}$, a strategy that is derived from Algorithm~\ref{alg:distMin}.
Note that, both of these strategies are optimal for  the probability to reach the target. However, the first strategy does not focus on optimizing the length to reach the target.
For both of these strategies, we calculate the expected conditional distance to the target in their induced Markov chains.
Table~\ref{table:10layouts} reports on our experimental results for a representative subset of the $100$ layouts we generated,
one of each decile
(one layout from the $10$ best percents, one from the $10-20\%$ range, etc).
\begin{table}[t]
	\centering
\begin{tabular}{|c||c|c|c|c|}
\hline
layouts ($\M$) & $\Val_{\M}(s_0, \Diamond T)$ & Shortest distance & $v_{DistOpt}$ & $v_{Storm}$\\
\hline
\hline
1 & 0.66 & 9 & 76.48 & 76.48 \\ %
2 & 0.52 & 18 & 299.75 & 629.16 \\ %
3 & 1.00 & 2 & 2.40 & 12.12 \\ %
4 & 1.00 & 3 & 3.44 & 34.47 \\ %
5 & 1.00 & 6 & 7.71 & 137.56 \\ %
6 & 0.68 & 10 & 264.04 & 9598.81 \\ %
7 & 1.00 & 5 & 112.69 & 9367.02 \\ %
8 & 0.91 & 10 & 11.49 & 5879.63 \\ %
9 & 1.00 & 3 & 3.66 & 5711.76 \\ %
10 & 0.91 & 5 & 12.89 & 149357.57 \\ %
\hline
\end{tabular} 	\caption{Comparison of the expected conditional length to reach the target for the strategies given by Algorithm~\ref{alg:distMin} ($v_{DistOpt}$) and {\sc Storm} ($v_{Storm}$) on some of the randomly generated layouts, sorted by their ratio.
 `Shortest distance' refer to the length of the shortest path to the target (without considering the stochastic dynamics of Frozen Lake) and `$\Val_{\M}(s_0, \Diamond T)$' represents the maximum probability of reaching the target from the initial position of the robot ($s_0$).}
	\label{table:10layouts}
\end{table}

Observe that the strategy given by Algorithm~\ref{alg:distMin} does not necessarily suggest following the shortest path, as this may not optimize the first objective (reaching the target with maximum probability).
For example, in the layout in Figure~\ref{fig:frozenLake}, the `shortest' path to the target has length $10$. But if we need to maximize the probability to reach the target, from the cell in the grid marked with~$1$, instead of going right, a better strategy would be to keep going to the cell above and then coming back. This way, the agent will avoid the hole below with certainty, and will eventually go to the right. This is the strategy that Algorithm~\ref{alg:distMin} provides, which has expected conditional length to the target $33.85$.
On the other hand, the expected conditional length to the target while following the optimal strategy produced by {\sc Storm} is much larger ($345.34$). This is because it asks the robot to loop in the $6\times 3$ area in the left. Because of the stochastic dynamics it eventually leaves this area and reaches the target, but it may take a long time, increasing the expected conditional length.

While performing the experiments on the $100$ randomly generated layouts, we observed that
in $9$ layouts out of $10$, %
the expected conditional length ($v_{\mathrm{Storm}}$) for the strategy $\sigma_{\mathrm{Storm}}$ is at least twice the expected conditional length ($v_{\mathrm{DistOpt}}$) for the strategy $\sigma_{\mathrm{DistOpt}}$.
In $69\%$ of the layouts, $v_{\mathrm{Storm}}$ values are $10$ times worse than the $v_{\mathrm{DistOpt}}$ values.
In the worst cases ($23\%$ of the layouts), $v_{\mathrm{Storm}}$ values are at least a $1000$ times worse than the $v_{\mathrm{DistOpt}}$ values.

\section{Safety and Expected Mean Payoff}
\label{sec:safety-expected-mp}

In this section, we consider another multi-objective problem -- as a first objective, we maximize the probability of avoiding a set of states in an MDP, and as a second objective, we maximize the expected conditional Mean Payoff. We propose a pruning-based algorithm, similar to Algorithm~\ref{alg:distMin}, to solve this problem.

For this section, we augment
the definition of an MDP $\M$  with a \emph{reward} function $R:S\times A\to \R$, where $S, A$ and $P$ are the same as in the previous sections. Furthermore, we consider a set of states $\Bad \subset S$ in $\M$ and assume that every state in $\Bad$ is a sink state.%

For an MDP $\M = (S,A,P,R)$,
let $\Pr_{\M_{\sigma},s}(\Box \neg \Bad)$ be the probability of avoiding all states in $\Bad$, starting from $s \in S$, following the strategy $\sigma$ in $\M$.
Then, let $\Val_\M(s) = \max_{\sigma}\Pr_{\M_{\sigma},s}(\Box \neg \Bad)$ be the maximum probability to avoid $\Bad$ from $s$, and $\Sigma_{\M,s}(\Box \neg \Bad) = \argmax_{\sigma}\Pr_{\M_{\sigma},s}(\Box\neg\Bad)$ be the set of all optimal strategies for safety.

To formally define the second objective, %
we first define the \emph{total reward} of horizon $n$ for a path $\rho = s_0a_0\ldots $ %
as $\Reward_n(\rho)=\sum_{i=0}^{n-1}R(s_i,a_i)$.
Then, for a strategy $\sigma$ and a state $s$, the \emph{expected mean-payoff} is defined as
$$\E_{\sigma,s}(\MP) = \liminf_{n\to\infty} \frac{1}{n} \E_{\sigma,s}(\Reward_n)\,.$$

The optimal \emph{expected average reward} starting from a state $s$ in an MDP $\M$
is defined over all strategies $\sigma$ in $\M$ as
$\sup_{\sigma}\E_{\sigma,s}(\MP)$.
One can restrict the supremum to the deterministic memoryless
strategies~\cite[section 9.1.4]{DBLP:books/wi/Puterman94}.

We  use $\E_{\sigma,s}(\Reward_n\mid \Box\neg\Bad)$ to denote the \emph{expected conditional finite horizon reward}.
Then the \emph{expected conditional mean-payoff} is defined as
$$\E_{\sigma,s}(\MP\mid \Box\neg \Bad) = \liminf_{n\to\infty}\frac{1}{n}\E_{\sigma,s}(\Reward_n\mid \Box\neg\Bad)\,.$$
Intuitively, it represents the expected mean-payoff one would obtain by following the strategy $\sigma$ and staying safe.

\subsection*{Problem statement.}

Given an MDP $\M$, an initial state $s_0$ and a set of states $\Bad$ where ${\Val_\M(s_0)>0}$,
our objective is to find a strategy that maximizes  $\E_{\M_{\sigma},s_0}(\MP \mid \Box\neg\Bad)$ among the strategies in $\Sigma_{\M,s_0}(\Box \neg \Bad)$, i.e., the strategies maximizing $\Pr_{\M_{\sigma},s_0}(\Box \neg \Bad)$.

For the rest of this section, we fix the MDP $\M = (S, A, P, R)$ and a set of bad states $\Bad \subset S$.
Note that, in this case, the functions $\sigma\mapsto \Pr_{{\sigma},s_0}(\Box \neg \Bad)$ and $\sigma\mapsto \E_{{\sigma},s_0}(\MP \mid \Box\neg\Bad)$ correspond to the two functions $f$ and $g$, respectively, and $\Sigma_{\M,s_0}(\Box \neg \Bad)$ corresponds to $\Sigma_f$, described in the introduction (Eq.~\ref{eq:conds}).

\subsection{Maximizing Probability of Staying Safe}

We denote the set $\left\{(s,a)\in S\times A\mid \Val_{\M}(s) = \sum_{s'}P(s,a,s')\cdot \Val_{\M}(s')\right\}$ using $\Opt_{\M}$. %
  For $s\in S$, let $\Opt_{\M}(s)$ be the set $\{a \mid (s,a)\in \Opt_{\M}\}$.
Finally, we use $\Sigma^{\Opt}_{\M}$ to represent the set of strategies that takes actions according to $\Opt_{\M}$, that is,
$\Sigma^{\Opt}_{\M} = \{\sigma\mid \forall \rho, \forall a\in\Supp(\sigma(\rho));(\last(\rho),a)\in \Opt_{\M}\}\,.$

We first state the following results, analogous to Lemma~\ref{lemma:case1} and~\ref{lem:localOptCoverReach} respectively, which can be proved similarly as in the case of reachability.
\begin{lemma}\label{lemma:case1Safety}
    For every state $s \in S\setminus \Bad$ and for every action $a$,
	\[\Val_{\M}(s)\geq \sum_{s'}P(s,a,s')\cdot \Val_{\M}(s')\,.
 \]
\end{lemma}

\begin{lemma}\label{lem:localOptSafety}
	For every state $s \in S$, $\Sigma_{\M,s}(\Box \neg \Bad)\subseteq \Sigma^{\Opt}_{\M}$.
\end{lemma}

Furthermore, we will show that, unlike reachability, in this case, the other direction of the containment also holds:
\begin{lemma}\label{lem:localOptSafety-1}
	For every state $s \in S$, $\Sigma_{\M,s}(\Box \neg \Bad)\supseteq \Sigma^{\Opt}_{\M}$.
\end{lemma}

In order to prove Lemma~\ref{lem:localOptSafety-1}, we first develop a few intermediate results. We start with defining the following notations:
\begin{align*}
	\UPre^0(\Bad) &= \Bad, \quad
	\UPre^{i+1}(\Bad) = \{s\mid \forall a, \exists s'\in \UPre^i(\Bad), P(s,a,s')> 0\},\\
	\UPre^*(\Bad) &= \bigcup_{i=0}^{\infty}\UPre^i(\Bad)\,.
\end{align*}
Furthermore, we define $\Good = S\setminus \UPre^*(\Bad)$, $V = S\setminus(\Good\cup \Bad)$.

\begin{lemma}\label{lem:safetyGood1}
	For every state $s \in S$, $\Val_{\M}(s) = 1$ iff $s\in\Good$.
\end{lemma}
\begin{proof}
		For $s\in \Good$, $\exists a$ such that $ \Supp(P(s,a))\subseteq \Good$. This gives a strategy to surely avoid $\Bad$, and hence $\Val_{\M}(s) = 1$.

		If $s \not\in \Good$, then either (i) $s \in \Bad$, in which case $\Val_{\M}(s) = 0$, or (ii) $s \in \UPre^*(\Bad)$, and hence $s\in \UPre^i(\Bad)\setminus \UPre^{i-1}(\Bad)$ for some $i$. Then, for every action $a$, $\Supp(P(s,a))\cap \UPre^{i-1}(\Bad) \neq \emptyset $. This implies, for any strategy $\sigma$, there is a path from $s$ of length at most $i$ reaching $\Bad$ following $\sigma$. Since this path has a non-zero probability, we therefore get that $\Val_{\M}(s) < 1$.
\end{proof}
For a strategy $\sigma$ and a finite path $\rho$, we define the strategy $\sigma_{\rho}$ as follows:
for any finite path $\rho'$ starting from $\last(\rho)$, $\sigma_\rho(\rho')=\sigma(\rho\cdot \rho')$.

\begin{lemma}\label{lem:safetyGood2}
	For every strategy $\sigma \in \Sigma^{\Opt}_{\M}$, and every finite path  $\rho$  in $\M$ following $\sigma$,
	$\Pr_{\sigma_\rho, \last(\rho)}(\Box \neg \Bad) = 1$ iff $\last(\rho) \in\Good$.
\end{lemma}
\begin{proof}
	We denote $\last(\rho)$ by $s$.
	First, let $s\in \Good$.
 Then we can show that $\forall a\in \Opt_{\M}(s), \Supp(P(s,a))\subseteq \Good$.
 Indeed, if there exists an action $a \in\Opt_{\M}(s)$ and a state $s' \in \Supp(P(s,a))$ such that $s' \notin \Good$, then from Lemma~\ref{lem:safetyGood1}, $\Val(s')<1$, which would further imply that
 \[\Val_{\M}(s) = \sum_{s'}P(s,a,s') \cdot \Val_{\M}(s') < \sum_{s'}P(s,a,s') = 1\,,
 \]
 which contradicts the fact that $s \in \Good$ (using Lemma~\ref{lem:safetyGood1}).
	So for every strategy $\sigma$ in $\Sigma^{\Opt}_{\M}$, every path from $s$ following $\sigma_\rho$  only visits states from $\Good$.
	Therefore, $\Pr_{\sigma_\rho,s}(\Box\neg\Bad)=1$.

	To conclude, observe that if $s \in S \setminus \Good$,
	$\Pr_{\sigma_\rho, s}(\Box \neg \Bad) \leq \Val_\M(s) < 1$.
\end{proof}

In the following, for the ease of notation,  for any state $s\in S$ and a strategy $\sigma$, we denote $\Pr_{\sigma,s}(\Cyl_{\sigma}(\rho))$ by $\PP_{\sigma,s}(\rho)$.
Recall that,  for every action $a\in\Opt_{\M}(s)$,
$\Val_{\M}(s) = \sum_{s'}P(s,a,s')\cdot \Val_{\M}(s')\,.$
We can then expand $\Val_{\M}(s)$ as:
$$\Val_{\M}(s) = \sum_a\sigma(s,a)\Val_{\M}(s) = \sum_a\sigma(s,a)\sum_{s'}P(s,a,s')\cdot \Val_{\M}(s')\,.$$
We can generalize the above statement by unfolding $\Val_\M(\cdot)$ for $n$ steps:

\begin{lemma}\label{lem:safety1}
	For every state $s\in S$ and for every strategy $\sigma\in\Sigma^{\Opt}_{\M}$,
	$$\Val_{\M}(s) = \sum_{\rho\in (VA)^{n}V}\PP_{\sigma,s}(\rho)\cdot \Val_{\M}(\last(\rho))+\sum_{\rho\in (VA)^{< n}\Good}\PP_{\sigma,s}(\rho)$$
\end{lemma}
The summation in the first term of the above expression is taken over all paths that reach neither $\Good$ nor $\Bad$ within $n$ steps, whereas the summation in the second term  is over all paths that reach some state in $\Good$ within $n$ steps.
The result in Lemma~\ref{lem:safety1} follows from the following result:

\begin{lemma}\label{lem:safety1-a}
	For every finite path $\rho = s_0a_0s_1\ldots s_n$ of length $n$ and for every strategy $\sigma$ in $\Sigma^{\Opt}_{\M}$, for all $k< n$:
	\small
	\[
 \Val_{\M}(s_k) = \sum_{\rho'\in (VA)^{n-k}V}\PP_{\sigma_{\rho_{|_k}},s_k}(\rho')\cdot \Val_{\M}(\last(\rho'))+\sum_{\rho'\in (VA)^{< n-k}\Good}\PP_{\sigma_{\rho_{|_k}},s_k}(\rho')\,.
 \]
	\normalsize
\end{lemma}

\noindent Using Lemma~\ref{lem:safety1-a}, we can now prove Lemma~\ref{lem:safety1}:
\begin{proof}[Proof of Lemma~\ref{lem:safety1}]
	Putting $k = 0$ in Lemma~\ref{lem:safety1-a}, we get:
	\small
	\[\Val_{\M}(s_0) = \sum_{\rho'\in (VA)^{n}V}\PP_{\sigma,s_0}(\rho')\cdot \Val_{\M}(\last(\rho'))+\sum_{\rho'\in (VA)^{< n}\Good}\PP_{\sigma,s_0}(\rho')\,.\qedhere\]
	\normalsize
\end{proof}
We now  characterize $\Pr{(\cdot)}$ in the same way as we did for $\Val(\cdot)$.
 Note that,
for any state $s$ and any strategy $\sigma$, we can expand $\Pr_{\sigma,s}(\Box \neg \Bad)$ as
$$\Pr_{\sigma,s}(\Box \neg \Bad) = \sum_a\sigma(s,a)\sum_{s'}P(s,a,s')\cdot \Pr_{\sigma_{sas'},s'}(\Box \neg \Bad)\,.$$
Analogous to Lemma~\ref{lem:safety1},  %
we can generalize this statement by unfolding $\Pr(\cdot)$ for $n$ steps:

\begin{lemma}\label{lem:safety2}
	For every state $s\in S$ and for every strategy $\sigma\in\Sigma^{\Opt}_{\M}$,
	\[\Pr_{\sigma,s}(\Box \neg \Bad) = \sum_{\rho\in (VA)^{n}V}\PP_{\sigma,s}(\rho)\cdot \Pr_{\sigma_{\rho},\last(\rho)}(\Box\neg\Bad)+\sum_{\rho\in (VA)^{< n}\Good}\PP_{\sigma,s}(\rho)\,.
 \]
\end{lemma}

\begin{lemma}\label{lem:safetyOpt}
	For every $s\in S$, and every $\sigma\in\Sigma^{\Opt}_{\M}$, $\Val_{\M}(s) = \Pr_{\sigma, s}(\Box \neg\Bad)$.
\end{lemma}
\begin{proof}
	If $s\in \Good$, $\Val_{\M}(s) = \Pr_{\sigma, s}(\Box \neg\Bad) =1$.
	If $s\in \Bad$,  $\Val_{\M}(s) = \Pr_{\sigma, s}(\Box \neg\Bad) = 0$.
	Finally, if $s\in V$, from Lemma~\ref{lem:safety1} and Lemma~\ref{lem:safety2},
 \small
	\begin{align*}
		\Val_{\M}(s) - \Pr_{\sigma, s}(\Box \neg\Bad) &= \sum_{\rho\in (VA)^{n}V}\PP_{\sigma,s}(\rho)\cdot (\Val_{\M}(\last(\rho))-\Pr_{\sigma_\rho, \last(\rho)}(\Box \neg \Bad))\\
		&< \sum_{\rho\in (VA)^{n}V}\PP_{\sigma,s}(\rho)\, \text{ [Using Lemma~\ref{lem:safetyGood2}]}
	\end{align*}
 \normalsize
For $s\in \UPre^*(\Bad)$, there is a path of length at most $|V|$ reaching $\Bad$ in $\M_{\sigma}$. So $\lim_{n\to \infty}\sum_{\rho\in (VA)^{n}V}\PP_{\sigma,s}(\rho)=0$.
\end{proof}

Lemma~\ref{lem:localOptSafety-1} follows directly from Lemma~\ref{lem:safetyOpt}.
Then, using Lemma~\ref{lem:localOptSafety} and~\ref{lem:localOptSafety-1}, we conclude the following theorem:

\begin{theorem}\label{thm:safetyOpt1}
  For every state $s \in S$, $\Sigma_{\M,s}(\Box\neg\Bad)= \Sigma^{\Opt}_{\M}$.
\end{theorem}

\subsection{Maximizing Expected Conditional Mean Payoff}
We propose a simple two-step pruning
algorithm, similar to Algorithm~\ref{alg:distMin}, to solve the multi-objective problem defined by safety and mean-payoff. We first modify the given MDP $\M$ in the following manner.

\begin{definition}\label{def:mprime-from-m-safety}
	Let $S' = \{s\in S\mid \Val_{\M}(s)> 0\}$. We define $\M' = (S',A,P',R)$ where
	 $P'$ is defined as follows:
	\begin{align*}
		P'(s,a,s')=&\begin{cases}
			P(s,a,s')\cdot\frac{\Val_{\M}(s')}{\Val_{\M}(s)} &\text{if } (s,a) \in \Opt_{\M}\text{ and }s\in S'\\
			\bot &\text{otherwise.}
		\end{cases}
	\end{align*}
\end{definition}

Note that $\M'$ is again well-defined. We now present the two-step algorithm:

\begin{algorithm}[ht]
	\caption{}\label{alg:expMPMax}
	\begin{algorithmic}
		\Require $\M = (S,A,P,R)$, $s_0 \in S$, $\Bad\subseteq S$
	\end{algorithmic}
	\begin{algorithmic}[1]
		\State Create the MDP $\M' = (S',A,P',R')$ according to Definition~\ref{def:mprime-from-m-safety}.
		\State Find a strategy $\sigma^*$ that maximizes the expected mean payoff in $\M'$:
		\[\textstyle \sigma^* \in \argmax_\sigma \E'_{\sigma,s_0}(\MP)\,.\]\\
		\Return $\sigma^*$.
	\end{algorithmic}
\end{algorithm}

For a state $s_0$, a strategy $\sigma$, and a finite path $\rho = s_0a_0s_1\ldots s_n
\in (S'A)^*S'\cap \FPaths_{\M'}(s_0,\sigma)$, we define,
$\GoodCyl_{\sigma}(\rho) = \Cyl_{\sigma}(\rho)\cap \{\rho'\mid \rho' \models \Box\neg\Bad\}$. %
Then, using Lemma~\ref{lem:safetyOpt}, we get that if $\sigma \in \Sigma_{\M}^{\Opt}$, then
\small
\begin{equation}\label{eq:goodCyl}
\Pr_{\sigma,s_0}(\GoodCyl_{\sigma}(\rho)) = \Pr_{\sigma,s_0}(\Cyl_{\sigma}(\rho))\cdot \Pr_{\sigma_{\rho},s_n}(\Box\neg\Bad)
= \Pr_{\sigma,s_0}(\Cyl_{\sigma}(\rho))\cdot \Val_{\M}(s_n).\end{equation}
\normalsize

\begin{lemma}\label{lemma:cylRelationSafety}
	For every strategy
	$\sigma\in\Sigma^{\Opt}_{\M}$, $s_0\in S'$ and every finite path
	$\rho=s_0a_0s_1\ldots s_n \in (S'A)^*S'\cap \FPaths_{\M'}(s_0,\sigma)$,
	$\Pr'_{\sigma,s_0}(\Cyl'_{\sigma}(\rho)) = \frac{\Pr_{\sigma,s_0}(\GoodCyl_{\sigma}(\rho))}{\Val_{\M}(s_0)}\,.$
\end{lemma}
\begin{proof}
    \begin{align*}
		\Pr'_{\sigma,s_0}(\Cyl'_{\sigma}(\rho))
		& = \prod_{i = 0}^{n-1} \sigma(\rho_{|i},a_i)\cdot P'(s_i,a_i,s_{i+1}) \\
		&= \prod_{i = 0}^{n-1} \sigma(\rho_{|i},a_i)\cdot P(s_i,a_i,s_{i+1})\cdot \frac{\Val_{\M}(s_{i+1})}{\Val_{\M}(s_i)}\\
		& = \Pr_{\sigma,s_0}(\Cyl_{\sigma}(\rho))\cdot \frac{\Val_{\M}(s_n)}{\Val_{\M}(s_0)}
        = \frac{\Pr_{\sigma,s_0}(\GoodCyl_{\sigma}(\rho))}{\Val_{\M}(s_0)} \quad[\text{from Eq.~\ref{eq:goodCyl}}]
        \qedhere
	\end{align*}
\end{proof}

We now show the following correlation between the expected mean-payoff in $\M'$ and the expected conditional mean-payoff in $\M$:

\begin{lemma}\label{lemma:expMaxMP}
	For every strategy $\sigma$, $\E'_{\sigma}(\MP) = {\E_{\sigma}(\MP\mid \Box\neg \Bad)}$
\end{lemma}
\begin{proof}
For $r\in \R$, we define $\xi_r=\{\rho\in \FPaths_\M (s_0)\cap (SA)^nS \mid \Reward_n(\rho)=r\}$ and $\xi'_r=\{\rho\in \FPaths_{\M'} (s_0)\cap (S'A)^nS' \mid \Reward_n(\rho)=r\}$. Note that for a fixed $n$, there are finitely many such non-empty $\xi_r$.
From the definition of the conditional expected reward in $\M$, we get:
\begin{align*}
\E_{\sigma,s_0}(\Reward_n \mid \Box \neg \Bad)    &=
\sum_{r} {r \cdot \frac{\Pr_{\sigma,s_0}(\{\rho \mid \Reward_n(\rho) = r\} \cap \Box \neg \Bad)}{\Pr_{\sigma,s_0}(\Box \neg \Bad)}}\\
&= \sum_{r} {r \cdot \sum_{\rho \in \xi_r} {\frac{\Pr_{\sigma,s_0}(\Cyl_\sigma(\rho) \cap \Box \neg \Bad)}{\Val_\M(s_0)}}}\\
&= \sum_{r} {r \cdot \sum_{\rho \in \xi_r} {\frac{\Pr_{\sigma,s_0}(\GoodCyl_\sigma(\rho))}{\Val_\M(s_0)}}}\,.
\end{align*}
\begin{align}
\text{Then, }\E'_{\sigma,s_0}(\Reward_n)
 &=\sum_{r} {r \cdot \Pr'_{\sigma,s_0}(\{\rho \mid \Reward_n(\rho) = r\})}
    = \sum_{r} {r \cdot \sum_{\rho \in \xi'_r} {\Pr'_{\sigma,s_0}(\Cyl'_\sigma(\rho))}} \nonumber\\
    &= \sum_{r} {r \cdot \sum_{\rho \in \xi'_r} {\frac{\Pr_{\sigma,s_0}(\GoodCyl_\sigma(\rho))}{\Val_\M(s_0)}}} \nonumber\\
    &= \sum_{r} {r \cdot \sum_{\rho \in \xi_r} {\frac{\Pr_{\sigma,s_0}(\GoodCyl_\sigma(\rho))}{\Val_\M(s_0)}}}  \label{eq:expectation}\\
    &= {\E_{\sigma,s_0}(\Reward_n\mid \Box\neg\Bad)} \label{eq:expectation-2}
\end{align}

The equality in Eq.~\ref{eq:expectation} is due to the fact that for any finite path $\rho = s_0 \ldots s_n \in (SA)^nS\setminus (S'A)^nS'$, $\exists i$ s.t. $\Val_\M(s_i) = 0$, which implies,  $\Pr_{\sigma,s_0}(\GoodCyl_\sigma(\rho)) \leq \Pr_{\sigma,s_0}(\GoodCyl_\sigma(s_0\ldots s_i))= \Pr_{\sigma,s_0\ldots s_i}(\Box\neg\Bad) \leq \Val_\M(s_i) = 0$.

Finally, dividing by $n$ and taking limit on the both sides of Eq.~\ref{eq:expectation-2}, we get $\E'_{\sigma,s_0}(\MP) = {\E_{\sigma,s_0}(\MP\mid \Box\neg \Bad)}$.
\end{proof}

Now we prove the correctness of Algorithm~\ref{alg:expMPMax}:

\begin{theorem}
\label{thm:3}
	 Given an MDP $\M = (S, A, P, R)$, a state $s_0 \in S$ and $\Bad \subset S$,
 let $\sigma^*$ be the strategy returned by Algorithm~\ref{alg:expMPMax}.
 Then,
	\begin{enumerate}
		\item $\Pr_{\sigma^*,s_0}(\Box\neg\Bad) = \Val_{\M}(s_0)$.
		\item $\displaystyle \E_{\sigma^*,s_0}(\MP\mid\Box\neg\Bad) =
		\max_{\sigma\in \Sigma_{\M,s_0}(\Box\neg\Bad)}\E_{\sigma,s_0}(\MP\mid\Box\neg\Bad)$
	\end{enumerate}
\end{theorem}

\begin{proof}
    	From Theorem~\ref{thm:safetyOpt1}, for any $\sigma$ in $\Sigma^{\Opt}_{\M}$,  $\Pr_{\sigma,s_0}(\Box\neg\Bad) = \Val_{\M}(s_0)$. Note that a strategy in $\M'$ would be in $\Sigma^{\Opt}_{\M}$. Therefore,
$\Pr_{\sigma^*,s_0}(\Box\neg\Bad) = \Val_{\M}(s_0)$.

    From Lemma~\ref{lemma:expMaxMP}, we get  for any $\sigma$,
	\begin{align*}
 \E_{\sigma,s_0}(\MP\mid \Box\neg \Bad)
	&=\E'_{\sigma,s_0}(\MP)\\
 \Rightarrow\argmax_{\sigma\in \Sigma_{\M,s_0}(\Box\neg \Bad)} \E_{\sigma,s_0}(\MP\mid \Box\neg \Bad)
	&=\argmax_{\sigma\in \Sigma_{\M,s_0}(\Box\neg \Bad)} \E'_{\sigma,s_0}(\MP)%
\end{align*}
Hence,  $\sigma^*\in \argmax_{\sigma\in \Sigma_{\M,s_0}(\Box\neg \Bad)}\E_{\sigma,s_0}(\MP\mid \Box\neg \Bad)$ and therefore,
we conclude,  $\displaystyle \E_{\sigma^*,s_0}(\MP\mid\Box\neg\Bad) =
\max_{\sigma\in \Sigma_{\M,s_0}(\Box\neg\Bad)}\E_{\sigma,s_0}(\MP\mid\Box\neg\Bad)$.
\end{proof}
Note that, constructing the MDP $\M'$ (Line~1 of Algorithm~\ref{alg:expMPMax}) takes polynomial time. Finding a strategy that optimizes $\E'_{\sigma,s_0}(\MP)$ also takes polynomial time~\cite[Chapter 9]{DBLP:books/wi/Puterman94}. Therefore, the overall algorithm %
takes polynomial time.

\section{Discussion}
\label{sec:discussion}

The work presented in this article proposes a pruning-based approach
(Algorithms~\ref{alg:distMin},~\ref{alg:expMPMax}) that can be used to solve
certain multi-objective problems in MDPs. The algorithms work by first pruning
the given MDP based on the first objective, and then solving the (possibly
simplified) second objective on the pruned MDP. Note that, optimizing the second objective, in turn, optimizes  both of the objectives in the lexicographic order.

The case where the first objective is to maximize the probability of reaching
a set of (target) states in an MDP and the second objective is to minimize the
conditional expected time to reach the same set of states, has been
discussed in Section~\ref{sec:reachability-expected-length}. Note that one can
consider more general (positive) cost functions and try to minimize the
conditional expected cost to reach the target states as a secondary objective,
keeping the first objective unchanged.

Based on a suggestion by Jakob Piribauer, we conjecture that the second objective considered in this paper can, in fact, be replaced by any measurable function.
More precisely, when the first objective is to remain safe, our technique can be applied to solve the bi-objective problem where the second objective is to optimize the expected value of a measurable function $g$, conditioned on the event that safety is satisfied. To this end, we can prove the following result:
every strategy in the original MDP $\M$ that maximizes $\mathbb{E}(g \mid \Box \neg {\sf Bad})$ while maximizing the probability of staying safe, also maximizes the expected value of $g$ in the pruned MDP $\M'$, that is,
\[\sup_{\sigma \in \Sigma_{\sf Safe}^\M} \mathbb{E}_{\sigma}(g \mid \Box \neg {\sf Bad}) = \sup_{\sigma \in \Sigma(\M')} \mathbb{E}_{\sigma}'(g)\]
where $\Sigma_{\sf Safe}^\M$ denotes the set of all strategies that maximize the probability of staying safe in $\M$. We believe this result can be proved by generalizing the proof of Lemma~\ref{lemma:expMaxMP}.

Similarly, when the primary objective is to reach a set of target states with as high probability as possible, we believe our technique will be able to compute the optimal strategy when the secondary objective is given by any measurable function~$g$. We conjecture that the following result will hold:
any strategy in $\M$ that first maximizes the probability to reach a target and further maximizes the expected value of a measurable function $g$ conditioned on reaching a target state, will also maximize the (unconditional) expected value of $g$ in the pruned MDP $\M'$, among the strategies that reach a target almost surely, that is, with probability $1$.
More formally, we can obtain the following result:
\[\sup_{\sigma \in \Sigma_{\sf Reach}^\M} \mathbb{E}_{\sigma}(g \mid \Diamond T) = \sup_{\sigma \in \Sigma^{\M'}_{\sf a.s.Reach}} \mathbb{E}_{\sigma}'(g)\]
where $\Sigma^{\M'}_{\sf a.s.Reach}$ is the set of all strategies in $\M'$ that, when followed, forces $\M'$ to reach a target state with probability~$1$.
Further, if it is the case that every strategy in $\M'$ maximizing the (unconditional) expected value of $g$ reaches a target with probability $1$ (which was the case in the pair of objectives considered in Section~\ref{sec:reachability-expected-length}), then the problem reduces to finding a strategy in $\M'$ that  maximizes the (unconditional) expected value of $g$ among all strategies, that is,
\[\sup_{\sigma \in \Sigma_{\sf Reach}^\M} \mathbb{E}_{\sigma}(g \mid \Diamond T) = \sup_{\sigma \in \Sigma(\M')} \mathbb{E}_{\sigma}'(g)\,.\]

While we studied only two-dimensional lexicographic objectives for
the sake of clarity and simplicity, we note that our work can be
straight-forwardly extended to more than two reward structures. For example,
one may want to optimize for safety first, reachability second, and minimal
expected time to reach a target as a third objective. In this case, we would
proceed in three steps: a first pruning of the MDP that solves the
safety problem, a second pruning that over-approximate the winning strategies
for reachability, and finally we would minimize the expected distance.

\bibliographystyle{splncs04}
\bibliography{main}

\begin{thebibliography}{10}
\providecommand{\url}[1]{\texttt{#1}}
\providecommand{\urlprefix}{URL }
\providecommand{\doi}[1]{https://doi.org/#1}

\bibitem{10.1007/978-3-642-54862-8_37}
Almagor, S., Boker, U., Kupferman, O.: Discounting in {LTL}. In:
  {\'A}brah{\'a}m, E., Havelund, K. (eds.) Tools and Algorithms for the
  Construction and Analysis of Systems. pp. 424--439. Springer Berlin
  Heidelberg, Berlin, Heidelberg (2014)

\bibitem{DBLP:conf/aaai/AlshiekhBEKNT18}
Alshiekh, M., Bloem, R., Ehlers, R., K{\"{o}}nighofer, B., Niekum, S., Topcu,
  U.: Safe reinforcement learning via shielding. In: Proceedings of the 32nd
  {AAAI} Conference on Artificial Intelligence, ({AAAI} 2018). pp. 2669--2678.
  {AAAI} Press (2018)

\bibitem{10.1007/978-3-642-54862-8_43}
Baier, C., Klein, J., Kl{\"u}ppelholz, S., M{\"a}rcker, S.: Computing
  conditional probabilities in {M}arkovian models efficiently. In:
  {\'A}brah{\'a}m, E., Havelund, K. (eds.) Tools and Algorithms for the
  Construction and Analysis of Systems. pp. 515--530. Springer Berlin
  Heidelberg, Berlin, Heidelberg (2014)

\bibitem{10.1007/978-3-662-54580-5_16}
Baier, C., Klein, J., Kl{\"u}ppelholz, S., Wunderlich, S.: Maximizing the
  conditional expected reward for reaching the goal. In: Legay, A., Margaria,
  T. (eds.) Tools and Algorithms for the Construction and Analysis of Systems.
  pp. 269--285. Springer Berlin Heidelberg, Berlin, Heidelberg (2017)

\bibitem{DBLP:journals/mor/BertsekasT91}
Bertsekas, D.P., Tsitsiklis, J.N.: An analysis of stochastic shortest path
  problems. Math. Oper. Res.  \textbf{16}(3),  580--595 (1991),
  \url{https://doi.org/10.1287/moor.16.3.580}

\bibitem{10.1007/978-3-642-36742-7_12}
Bohy, A., Bruy{\`e}re, V., Filiot, E., Raskin, J.F.: Synthesis from {LTL}
  specifications with mean-payoff objectives. In: Piterman, N., Smolka, S.A.
  (eds.) Tools and Algorithms for the Construction and Analysis of Systems. pp.
  169--184. Springer Berlin Heidelberg, Berlin, Heidelberg (2013)

\bibitem{AAMAS23}
Chakraborty, D., Busatto{-}Gaston, D., Raskin, J., P{\'{e}}rez, G.A.:
  Formally-sharp dagger for {MCTS:} lower-latency monte carlo tree search using
  data aggregation with formal methods. In: Agmon, N., An, B., Ricci, A., Yeoh,
  W. (eds.) Proceedings of the 2023 International Conference on Autonomous
  Agents and Multiagent Systems, {AAMAS} 2023, London, United Kingdom, 29 May
  2023 - 2 June 2023. pp. 1354--1362. {ACM} (2023),
  \url{https://dl.acm.org/doi/10.5555/3545946.3598783}

\bibitem{10.1007/978-3-642-19835-9_24}
Chatterjee, K., Henzinger, T.A., Jobstmann, B., Singh, R.: {QUASY}:
  Quantitative synthesis tool. In: Abdulla, P.A., Leino, K.R.M. (eds.) Tools
  and Algorithms for the Construction and Analysis of Systems. pp. 267--271.
  Springer Berlin Heidelberg, Berlin, Heidelberg (2011)

\bibitem{Chatterjee2023}
Chatterjee, K., Katoen, J.P., Mohr, S., Weininger, M., Winkler, T.: Stochastic
  games with lexicographic objectives. Formal Methods in System Design  (Mar
  2023), \url{https://doi.org/10.1007/s10703-023-00411-4}

\bibitem{DBLP:conf/stacs/ChatterjeeMH06}
Chatterjee, K., Majumdar, R., Henzinger, T.A.: Markov decision processes with
  multiple objectives. In: Durand, B., Thomas, W. (eds.) {STACS} 2006, 23rd
  Annual Symposium on Theoretical Aspects of Computer Science, Marseille,
  France, February 23-25, 2006, Proceedings. Lecture Notes in Computer Science,
  vol.~3884, pp. 325--336. Springer (2006),
  \url{https://doi.org/10.1007/11672142\_26}

\bibitem{DBLP:conf/aaai/Chatterjee0PRZ17}
Chatterjee, K., Novotn{\'{y}}, P., P{\'{e}}rez, G.A., Raskin, J., Zikelic, D.:
  Optimizing expectation with guarantees in {POMDP}s. In: Singh, S.,
  Markovitch, S. (eds.) Proceedings of the Thirty-First {AAAI} Conference on
  Artificial Intelligence, February 4-9, 2017, San Francisco, California,
  {USA}. pp. 3725--3732. {AAAI} Press (2017),
  \url{http://aaai.org/ocs/index.php/AAAI/AAAI17/paper/view/14354}

\bibitem{10.1007/978-3-642-40196-1_28}
Chen, T., Kwiatkowska, M., Simaitis, A., Wiltsche, C.: Synthesis for
  multi-objective stochastic games: An application to autonomous urban driving.
  In: Joshi, K., Siegle, M., Stoelinga, M., D'Argenio, P.R. (eds.) Quantitative
  Evaluation of Systems. pp. 322--337. Springer Berlin Heidelberg, Berlin,
  Heidelberg (2013)

\bibitem{dede287b-e396-3536-a645-5de6f58195ff}
Denardo, E.V.: Computing a bias-optimal policy in a discrete-time {M}arkov
  decision problem. Operations Research  \textbf{18}(2),  279--289 (1970),
  \url{http://www.jstor.org/stable/168684}

\bibitem{10.1007/978-3-642-19835-9_11}
Forejt, V., Kwiatkowska, M., Norman, G., Parker, D., Qu, H.: Quantitative
  multi-objective verification for probabilistic systems. In: Abdulla, P.A.,
  Leino, K.R.M. (eds.) Tools and Algorithms for the Construction and Analysis
  of Systems. pp. 112--127. Springer Berlin Heidelberg, Berlin, Heidelberg
  (2011)

\bibitem{10.1007/978-3-030-90870-6_8}
Hahn, E.M., Perez, M., Schewe, S., Somenzi, F., Trivedi, A., Wojtczak, D.:
  Model-free reinforcement learning for lexicographic omega-regular objectives.
  In: Huisman, M., P{\u{a}}s{\u{a}}reanu, C., Zhan, N. (eds.) Formal Methods.
  pp. 142--159. Springer International Publishing, Cham (2021)

\bibitem{STORM22}
Hensel, C., Junges, S., Katoen, J., Quatmann, T., Volk, M.: The probabilistic
  model checker {Storm}. Int. J. Softw. Tools Technol. Transf.  \textbf{24}(4),
   589--610 (2022), \url{https://doi.org/10.1007/s10009-021-00633-z}

\bibitem{DBLP:conf/tacas/Junges0DTK16}
Junges, S., Jansen, N., Dehnert, C., Topcu, U., Katoen, J.: Safety-constrained
  reinforcement learning for {MDP}s. In: Chechik, M., Raskin, J. (eds.) Tools
  and Algorithms for the Construction and Analysis of Systems - 22nd
  International Conference, {TACAS} 2016, Held as Part of the European Joint
  Conferences on Theory and Practice of Software, {ETAPS} 2016, Eindhoven, The
  Netherlands, April 2-8, 2016, Proceedings. Lecture Notes in Computer Science,
  vol.~9636, pp. 130--146. Springer (2016),
  \url{https://doi.org/10.1007/978-3-662-49674-9\_8}

\bibitem{KNP11}
Kwiatkowska, M., Norman, G., Parker, D.: {PRISM} 4.0: Verification of
  probabilistic real-time systems. In: Gopalakrishnan, G., Qadeer, S. (eds.)
  Proc. 23rd International Conference on Computer Aided Verification (CAV'11).
  LNCS, vol.~6806, pp. 585--591. Springer (2011)

\bibitem{DBLP:books/wi/Puterman94}
Puterman, M.L.: Markov Decision Processes: Discrete Stochastic Dynamic
  Programming. Wiley Series in Probability and Statistics, Wiley (1994).
  \doi{10.1002/9780470316887}

\bibitem{ijcai2022p476}
Skalse, J., Hammond, L., Griffin, C., Abate, A.: Lexicographic multi-objective
  reinforcement learning. In: Raedt, L.D. (ed.) Proceedings of the Thirty-First
  International Joint Conference on Artificial Intelligence, {IJCAI-22}. pp.
  3430--3436. International Joint Conferences on Artificial Intelligence
  Organization (7 2022), \url{https://doi.org/10.24963/ijcai.2022/476}, main
  Track

\end{thebibliography}

\newpage
\appendix
\section{Missing proofs of Section~\ref{sec:safety-expected-mp}}
\subsection{Proof of Lemma~\ref{lemma:case1Safety}}
	Suppose, there is a state $s\in S\setminus \Bad$ and an action $a\in A$ such that $\Val_{\M}(s)< \sum_{s'}P(s,a,s')\cdot \Val_{\M}(s')$.
Now, consider the strategy $\sigma'$ that takes action $a$ from $s$ and then from paths $s\cdot as'$ follows a strategy $\sigma_{s'}\in \Sigma_{\M,s'}(\Box\neg\Bad)$ that maximizes the probability to avoid states in $\Bad$. Formally,
\begin{align*}
	\sigma'(\rho)=&\begin{cases}
		a &\text{if } \rho = s\\
		\sigma_{s'}(\rho') &\text{if } \rho = s\cdot as'\cdot \rho'%
	\end{cases}
\end{align*}
Then, we get the following: $\Pr_{\M_{\sigma'},s}(\Box\neg\Bad) =
\sum_{s'}P(s,a,s')\cdot \Pr_{{\sigma_{s'}},s'}(\Box \neg \Bad) =  \sum_{s'}P(s,a,s')\cdot \Val_{\M}(s') > \Val_{\M}(s)$, which is a contradiction.

\subsection{Proof of Lemma~\ref{lem:localOptSafety}}
    Suppose that there is a strategy $\sigma^*\in \Sigma_{\M,s}(\Box\neg \Bad)$ where there exists a path $\rho$ and an action $a\in \Supp(\sigma^*(\rho))$ such that $(\last(\rho),a)\not\in \Opt_{\M}$. Let $\last(\rho) = t$. Assume that $\rho\models\Box\neg\Bad$.
Then, from Lemma~\ref{lemma:case1Safety} and the fact that $(t,a)\not\in \Opt_{\M}$,
$$\Val_{\M}(\last(\rho))> \sum_{s'}P(t,a,s')\cdot \Val_{\M}(s')$$
and for every other action $a'\neq a$,
$$\Val_{\M}(t)\geq \sum_{s'}P(t,a',s')\cdot \Val_{\M}(s')\,.$$
Consider the strategy $\sigma''$ which differs from $\sigma^*$ only on paths with $\rho$ as prefix: on every path having $\rho$ as a prefix, $\sigma''$ takes the next action according to a strategy $\sigma_{t}\in \Sigma_{\M,t}(\Box\neg\Bad)$ that maximizes the probability to avoid states in $\Bad$ from $\last(\rho)$, whereas, it takes action according to $\sigma^*$ on every other path. Formally,
\begin{align*}
	\sigma''(\rho')=&\begin{cases}
		\sigma_{t}(\rho'') &\text{if } \rho' = \rho\cdot \rho''\\
		\sigma^*(\rho') &\text{otherwise.}
	\end{cases}
\end{align*}

Note that, for every strategy $\sigma$,
$$\Pr_{\M_{\sigma},\rho}(\Box\neg\Bad)
= \sum_{a'}\bigg(\sigma(\rho,a')
\cdot \sum_{s'}\left(P(t, a',s')\cdot \Pr_{\M_{\sigma},\rho\cdot a's'}(\Box\neg\Bad)\right)\bigg)\,.$$
Also, $\Pr_{\M_{\sigma^*},\rho\cdot a's'}(\Box\neg\Bad) \leq \Val_{\M}(s')$ for all $a' \in A$. Therefore,
\begin{align*}
	\Pr_{\M_{\sigma^*},\rho}(\Box\neg\Bad)
	&= \sum_{a'}\bigg(\sigma^*(\rho,a')
	\cdot \sum_{s'}\left(P(t, a',s')\cdot \Pr_{\M_{\sigma^*},\rho\cdot a's'}(\Box\neg\Bad)\right)\bigg)\\
	&\leq \sum_{a'}\bigg(\sigma^*(\rho,a')
	\cdot \sum_{s'}\left(P(t, a',s')\cdot \Val_{\M}(s')\right)\bigg)\\
	&< \sum_{a'}\sigma^*(\rho,a')\cdot
	\Val_{\M}(t)\\
	&=\Val_{\M}(t)
\end{align*}
So, $\Pr_{\M_{\sigma''},\rho}(\Box\neg\Bad) = \Pr_{\M_{\sigma_{t}},t}(\Box\neg\Bad) = \Val_{\M}(t, T) > \Pr_{\M_{\sigma^*},\rho}(\Box\neg\Bad)$.
Now note that, for any strategy $\sigma$,
\begin{align*}
	\Pr_{\M_{\sigma},s}(\Box\neg\Bad) &= \Pr_{\M_{\sigma},s}(\rho'\models\Box\neg\Bad\wedge \rho\sqsubseteq \rho') + \Pr_{\M_{\sigma},s}(\rho'\models\Box\neg\Bad\wedge \rho\not\sqsubseteq \rho')\\
	&= \Pr_{\M_{\sigma},s}(\cyl_{\sigma}(\rho))\cdot \Pr_{\M_{\sigma},\rho}(\Box\neg\Bad) + \Pr_{\M_{\sigma},s}(\rho'\models\Box\neg\Bad\wedge \rho\not\sqsubseteq \rho')
\end{align*}
As $\sigma^*(\rho') = \sigma''(\rho')$ for any $\rho'$ such that $\rho \sqsubseteq \rho'$, $$\Pr_{\M_{\sigma^*},s}((\cyl_{\sigma^*}(\rho)))=\Pr_{\M_{\sigma''},s}((\cyl_{\sigma''}(\rho)))\,.$$
and
$$\Pr_{\M_{\sigma^*},s}(\rho'\models\Box\neg\Bad\wedge \rho\not\sqsubseteq \rho')=\Pr_{\M_{\sigma''},s}(\rho'\models\Box\neg\Bad\wedge \rho\not\sqsubseteq \rho')\,.$$
Then $\Pr_{\M_{\sigma^*},s}(\Box\neg\Bad)<\Pr_{\M_{\sigma''},s}(\Box\neg\Bad)$, which cannot be true as $\sigma^*$ is an optimal strategy.
Therefore, $\Sigma_{\M,s}(\Box\neg\Bad)\subseteq\Sigma^{\Opt}_{\M}$.

\subsection{Proof of Lemma~\ref{lem:safety1-a}}

	We prove this by backward induction on $k$.

 \noindent
 Base case:  $k = n -1$. If $s_{n-1} \in \Good$, then $\Val_{\M}(s_{n-1}) = 1$. If $s_{n-1} \in \Bad$, then $\Val_{\M}(s_{n-1}) = 0$. In both cases, the statement is trivially true. If $s_{n-1} \in V$,
	\small
	\begin{align*}
		\Val_{\M}(s_{n-1}) &= \sum_{a_{n-1}} \sigma(\rho_{|_{n-1}},a_{n-1})\sum_{s'_n}P(s_{n-1},a_{n-1},s'_n)\cdot \Val_{\M}(s'_n)\\
		&= \sum_{\rho'\in \{s_{n-1}\}AS}\PP_{\sigma_{\rho_{|_{n-1}}},s_{n-1}}(\rho')\cdot \Val_{\M}(\last(\rho'))\\
		&= \sum_{\rho'\in \{s_{n-1}\}AV}\PP_{\sigma_{\rho_{|_{n-1}}},s_{n-1}}(\rho')\cdot\Val_{\M}(\last(\rho')) + \\
        &\quad \sum_{\rho'\in \{s_{n-1}\}A\Good}\PP_{\sigma_{\rho_{|_{n-1}}},s_{n-1}}(\rho')\times 1  +
		\sum_{\rho'\in \{s_{n-1}\}A\Bad}\PP_{\sigma_{\rho_{|_{n-1}}},s_{n-1}}(\rho') \times 0\\
		&= \sum_{\rho'\in VAV}\PP_{\sigma_{\rho_{|_{n-1}}},s_{n-1}}(\rho')\cdot\Val_{\M}(\last(\rho'))+\sum_{\rho'\in VA\Good}\PP_{\sigma_{\rho_{|_{n-1}}},s_{n-1}}(\rho')
	\end{align*}
	\normalsize
	Suppose the statement is true for $k+1$. Then,
	\small
	\begin{align*}
		\Val_{\M}(s_k) &= \sum_{a_k} \sigma(\rho_{|_{k}},a_{k})\sum_{s'_{k+1}}P(s_{k},a_k,s'_{k+1}) \Val_{\M}(s'_{k+1})\\
		&= \sum_{a_k} \sigma(\rho_{|_k},a_k)\left(\sum_{s'_{k+1}\in V}P(s_k,a_k,s'_{k+1}) \Val_{\M}(s'_{k+1})+
		\right.\\
		& \left. \sum_{s'_{k+1}\in \Good}P(s_k,a_k,s'_{k+1}) \Val_{\M}(s'_{k+1})+\sum_{s'_{k+1}\in \Bad}P(s_k,a_k,s'_{k+1}) \Val_{\M}(s'_{k+1})\right)\\
		&= \sum_{a_k} \sigma(\rho_{|_k},a_k)\left(\sum_{s'_{k+1}\in V}P(s_k,a_k,s'_{k+1}) \Val_{\M}(s'_{k+1})+
		\sum_{s'_{k+1}\in \Good}P(s_k,a_k,s'_{k+1})\right)\\
		&= \sum_{a_k} \sigma(\rho_{|_k},a_k)\sum_{s'_{k+1}\in V}P(s_k,a_k,s'_{k+1})
		\sum_{\rho'\in (\{s'_{k+1}\}\cdot A)^{n-k-1}\cdot V}
		\PP_{\sigma_{\rho_{|_{k+1}}},s'_{k+1}}(\rho')\cdot \Val_{\M}(\last(\rho'))\\
		&+\sum_{a_k} \sigma(\rho_{|_k},a_k)\sum_{s'_{k+1}\in V}P(s_k,a_k,s'_{k+1})\sum_{\rho'\in (\{s'_{k+1}\}\cdot A)^{\leq n-k-1}\cdot \Good}\PP_{\sigma_{\rho_{|_{k+1}}},s'_{k+1}}(\rho')\\
		&+
		\sum_{a_k} \sigma(\rho_{|_k},a_k)\sum_{s'_{k+1}\in \Good}P(s_k,a_k,s'_{k+1})\\
		&=\sum_{\rho'\in (VA)^{n-k}V}\PP_{\sigma_{\rho_{|_{k}}},s_{k}}(\rho')\cdot \Val_{\M}(\last(\rho'))+\sum_{\rho'\in (VA)^{< n-k}\Good}\PP_{\sigma_{\rho_{|_{k}}},s_{k}}(\rho')\,.
		\qedhere
	\end{align*}

\subsection{Proof of Lemma~\ref{lem:safety2}}

Here, for the ease of notation, we will use $\Val_{\sigma}(s)$ to denote $\Pr_{\sigma,s}(\Box \neg \Bad)$.

\begin{lemma}\label{lem:safety2-a}
	Let $\rho = s_0a_0s_1\ldots s_n$ be a finite path of length $n$. Let $\sigma$ be a strategy in $\Sigma^{\Opt}_{\M}$. Then for all $k< n$:
	\small
	$$\Val_{\sigma_{p_{|_k}}}(s_k) = \sum_{\rho'\in (VA)^{n-k}V}\PP_{\sigma_{p_{|_k},s_k}}(\rho')\cdot \Val_{\sigma_{p_{|_k}\cdot \rho'}}(\last(\rho'))
	+\sum_{\rho'\in (VA)^{< n-k}\Good}\PP_{\sigma_{p_{|_k},s_k}}(\rho')$$
	\normalsize
\end{lemma}
\begin{proof}
	We prove this by backward induction on $k$.
	For $k = n -1$. If $s_{n-1} \in \Good$, then $\Val_{\M}(s_{n-1}) = 1$. If $s_{n-1} \in \Bad$, then $\Val_{\M}(s_{n-1}) = 0$. In both cases, the statement is trivially true. If $s_{n-1} \in V$,
	\small
	\begin{align*}
		\Val_{\sigma_{p_{|_{n-1}}}}(s_{n-1}) &= \sum_{a'_{n-1}} \sigma(p_{|_{n-1}},a'_{n-1})\sum_{s'_n}P(s_{n-1},a'_{n-1},s'_n) \Val_{\sigma_{p_{|_{n-1}}\cdot a'_{n-1}s'_{n-1}}}(s'_n)\\
		&= \sum_{\rho'\in \{s_{n-1}\}AS}\PP_{\sigma_{p_{|_{n-1}}},s_{n-1}}(\rho')\cdot \Val_{\sigma_{p_{|_{n-1}}\cdot \rho'}}(\last(\rho'))\\
		&= \sum_{\rho'\in \{s_{n-1}\}AV}\PP_{\sigma_{p_{|_{n-1}}},s_{n-1}}(\rho')\cdot\Val_{\sigma_{p_{|_{n-1}}\cdot \rho'}}(\last(\rho'))\\
        &\quad \quad \quad \quad +\sum_{\rho'\in \{s_{n-1}\}A\Good}\PP_{\sigma_{p_{|_{n-1}}},s_{n-1}}(\rho')\times 1 \\
		&\quad \quad \quad \quad +
		\sum_{\rho'\in \{s_{n-1}\}A\Bad}\PP_{\sigma_{p_{|_{n-1}}},s_{n-1}}(\rho') \times 0\\
		&= \sum_{\rho'\in VAV}\PP_{\sigma_{p_{|_{n-1}}},s_{n-1}}(\rho')\cdot\Val_{\sigma_{p_{|_{n-1}}\cdot \rho'}}(\last(\rho'))+\sum_{\rho'\in VA\Good}\PP_{\sigma_{p_{|_{n-1}}},s_{n-1}}(\rho')
	\end{align*}
	\normalsize
	Suppose the statement is true for $k+1$. Then, using Lemma~\ref{lem:safetyGood2},
	\small
	\begin{align*}
		\Val_{\sigma_{p_{|_k}}}(s_k) &= \sum_{a'_k} \sigma(p_{|_{k}},a'_{k})\sum_{s'_{k+1}}P(s_{k},a'_k,s'_{k+1}) \Val_{\sigma_{p_{|_k}\cdot a'_{k}s'_{k+1}}}(s'_{k+1})\\
		&= \sum_{a_k} \sigma(p_{|_k},a_k)\left(
		\sum_{s'_{k+1}\in V}P(s_{k},a'_k,s'_{k+1}) \Val_{\sigma_{p_{|_k}\cdot a'_{k}s'_{k+1}}}(s'_{k+1})\right.\\
		&\quad \quad \quad \quad +\left.
		\sum_{s'_{k+1}\in\Good}P(s_{k},a'_k,s'_{k+1}) \Val_{\sigma_{p_{|_k}\cdot a'_{k}s'_{k+1}}}(s'_{k+1})\right.\\
		&\quad \quad \quad \quad +\left.
		\sum_{s'_{k+1}\in\Bad}P(s_{k},a'_k,s'_{k+1}) \Val_{\sigma_{p_{|_k}\cdot a'_{k}s'_{k+1}}}(s'_{k+1})
		\right)\\
		&= \sum_{a_k} \sigma(p_{|_k},a_k)\left(
		\sum_{s'_{k+1}\in V}P(s_{k},a'_k,s'_{k+1}) \Val_{\sigma_{p_{|_k}\cdot a'_{k}s'_{k+1}}}(s'_{k+1})\right.\\
		&\quad + \left. \sum_{s'_{k+1}\in\Good}P(s_{k},a'_k,s'_{k+1})\right)\\
		&= \sum_{a_k} \sigma(p_{|_k},a_k)
		\sum_{s'_{k+1}\in V}P(s_{k},a'_k,s'_{k+1})\times\\
        &
		\quad \quad\left(
		\sum_{\rho'\in (VA)^{n-k-1}V}\PP_{\sigma_{p_{|_k}\cdot a'_{k}s'_{k+1},s_{k+1}}}(\rho')\cdot \Val_{\sigma_{p_{|_k}\cdot a'_{k}s'_{k+1}\cdot \rho'}}(\last(\rho')) \right.\\
        &\quad\quad\quad\left. + \sum_{\rho'\in (VA)^{< n-k}\Good}\PP_{\sigma_{p_{|_k},s_k}}(\rho')
		\right)\\
		&\quad \quad  +\sum_{a_k} \sigma(p_{|_k},a_k)\sum_{s'_{k+1}\in\Good}P(s_{k},a'_k,s'_{k+1})\\
		&= \sum_{\rho'\in (VA)^{n-k}V}\PP_{\sigma_{p_{|_k},s_k}}(\rho')\cdot \Val_{\sigma_{p_{|_k}\cdot \rho'}}(\last(\rho'))\\
		&\quad\quad +\sum_{\rho'\in (VA)^{< n-k}\Good}\PP_{\sigma_{p_{|_k},s_k}}(\rho')\qedhere
	\end{align*}
	\normalsize
\end{proof}
\begin{proof}[Proof of Lemma~\ref{lem:safety2}]
	For $k = 0$ in Lemma~\ref{lem:safety2-a}, we get:
	\small
	$$\Val_{\sigma}(s_0) = \sum_{p'\in (VA)^{n}V}\PP_{\sigma,s_0}(p')\cdot \Val_{\sigma}(\last(p'))+\sum_{p'\in (VA)^{< n}\Good}\PP_{\sigma,s_0}(p')\,.\qedhere$$
	\normalsize
\end{proof}

\end{document}